\documentclass[journal, a4paper]{IEEEtran}
\usepackage{graphicx}
\usepackage{amssymb}
\usepackage{amsmath}
\usepackage{mathtools}
\usepackage{cite}
\usepackage{subfigure}
\usepackage[displaymath,mathlines]{lineno}
\usepackage{color}
\usepackage{multirow}
\usepackage{algpseudocode}
\usepackage{algorithm,algpseudocode}
\usepackage{algorithmicx}
\usepackage{multicol}
\usepackage{relsize}
\usepackage{enumitem}
\usepackage{microtype}
\usepackage{iftex}
\usepackage{url}
\ifLuaTeX\usepackage[utf8]{luainputenc}\else\usepackage[utf8]{inputenc}\fi
\usepackage[OT1]{fontenc}%In order to use the obsolete font specification of IEEEtran with LuaLaTex in TexLive 2016 and newer

\newcommand{\doublewidetilde}[1]{{%
		\mathpalette\double@widetilde{#1}%
}}

\makeatletter

\def\BState{\State\hskip-\ALG@thistlm}
\makeatother

\usepackage{amsthm}
\newtheorem{definition}{Definition}
\newtheorem{theorem}{Theorem}
\newtheorem{lemma}{Lemma}
\newtheorem{corollary}{Corollary}

\newcommand{\tr}{\ensuremath{\mathsf{T}}}
\newcommand{\conjtr}{\ensuremath{\mathsf{H}}}
\newcommand{\Exp}{\ensuremath{\mathbb{E}}}
\newcommand{\abb}[1]{#1} %{\textsc{\MakeLowercase{#1}}}

\begin{document}
	\title{Large-Scale-Fading Decoding in Cellular Massive MIMO Systems with Spatially Correlated Channels}

	\author{Trinh~Van~Chien,~\IEEEmembership{Student Member,~IEEE,}
		Christopher~Moll\'{e}n,
		Emil~Bj\"{o}rnson,~\IEEEmembership{Senior Member,~IEEE}
		\thanks{This paper was supported by the European Union's Horizon 2020 research and innovation programme under grant agreement No 641985 (5Gwireless). It was also supported by ELLIIT and CENIIT. Parts of this paper were submitted to IEEE International Conference on Communications 2019 \cite{Chien2019a}.}
		\thanks{T. V. Chien and E. Bj\"{o}rnson are with the Department of Electrical Engineering (ISY), Link\"{o}ping University, 581~83 Link\"{o}ping, Sweden (e-mail: trinh.van.chien@liu.se, emil.bjornson@liu.se).}
		\thanks{C. Moll\'{e}n was with the Department of Electrical Engineering (ISY), Linköping University, 581~83 Linköping, Sweden, when this work was done (e-mail: chris.mollen@gmail.com).}
	}
	
	\maketitle
	
	\begin{abstract}
		Massive multiple-input--multiple-output (\abb{MIMO}) systems can suffer from coherent intercell interference due to the phenomenon of pilot contamination. This paper investigates a two-layer decoding method that mitigates both coherent and non-coherent interference in multi-cell Massive \abb{MIMO}. To this end, each base station (\abb{BS}) first estimates the channels to intra-cell users using either minimum mean-squared error (\abb{MMSE}) or element-wise MMSE (\abb{EW-MMSE}) estimation based on uplink pilots. The estimates are used for local decoding on each \abb{BS} followed by a second decoding layer where the \abb{BS}s cooperate to mitigate inter-cell interference. An uplink achievable spectral efficiency (\abb{SE}) expression is computed for arbitrary two-layer decoding schemes.  A closed-form expression is then obtained for correlated Rayleigh fading, maximum-ratio combining, and the proposed large-scale fading decoding (\abb{LSFD}) in the second layer.  We also formulate a sum \abb{SE} maximization problem with both the data power and \abb{LSFD} vectors as optimization variables.  Since this is an \abb{NP}-hard problem, we develop a low-complexity algorithm based on the weighted \abb{MMSE} approach to obtain a local optimum. The numerical results show that both data power control and \abb{LSFD} improve the sum \abb{SE} performance over single-layer decoding multi-cell Massive \abb{MIMO} systems.
	\end{abstract}
	
	\begin{IEEEkeywords}
		Massive MIMO, Large-Scale Fading Decoding, Sum Spectral Efficiency Optimization, Channel Estimation.
	\end{IEEEkeywords}% no keywords
	
	\IEEEpeerreviewmaketitle
	
	\section{Introduction}
	\IEEEPARstart{M}{assive} \abb{MIMO} \abb{BS}s, which are equipped with hundreds of antennas, exploit channel reciprocity to estimate the channel based on uplink pilots and spatially multiplex a large number of users on the same time--frequency resource \cite{Bjornson2016b, marzetta2010noncooperative}.  
	It is a promising technique to meet the growing demand for wireless data traffic of tomorrow \cite{Gupta2015survey, Andrews2014a}. In a single-cell scenario, there is no need for computationally heavy decoding or precoding methods in Massive \abb{MIMO}, such as successive interference cancellation or dirty paper coding. Linear processing schemes (e.g., zero-forcing combining) can effectively suppress interference and noise if the \abb{BS} is equipped with a large number of antennas \cite{Marzetta2016a}.
	In a multi-cell scenario, however, pilot-based channel estimation is contaminated by the non-orthogonal transmission in other cells.  This results in coherent intercell interference in the data transmission, so-called \emph{pilot contamination} \cite{Jose2011b}, unless some advanced processing schemes are used to suppress it \cite{bjornson2018a}. 
	Pilot contamination causes the gain of having more antennas to decrease and the \abb{SE} of linear decoding methods, such as maximum-ratio combining (\abb{MRC}) or zero-forcing, to saturate as the number of antennas grows. 
	
	Much work has been done to mitigate the effects of pilot contamination. 
	The first and intuitive approach to mitigate pilot contamination is to increase the length of the pilots.  In practical networks, however, it is not possible to make all pilots orthogonal due to the limited channel coherence block \cite{Bjornson2016a}.
	Hence, there is a trade-off between having longer pilots and low pilot overhead.
	Another method to mitigate pilot contamination is to assign the pilots in a way that reduces the contamination \cite{Jin2015a}, since only a few users from other cells cause substantial contamination. 
	The pilot assignment is a combinatorial problem and heuristic algorithms with low computational complexity can be developed to mitigate the pilot contamination.  
	In \cite{Xu2015a}, a greedy pilot assignment method is developed that exploits the statistical channel information and mutual interference between users.  
	Pilot assignment approaches still suffer from asymptotic \abb{SE} saturation since we only change one contaminating user for a less contaminating user.
	A third method is to utilize the spatial correlation to mitigate the coherent interference using multi-cell minimum-mean-square error (\abb{M-MMSE}) decoding   \cite{bjornson2018a}, but this method has high computational complexity.
	
	Instead of combating pilot contamination, one can utilize it using more advanced decoding schemes \cite{ashikhmin2012a, Ashikhmin2018a, adhikary2017a}.  
	This approach was initially called \textit{pilot contamination decoding} since the \abb{BS}s cooperate in suppressing the pilot contamination \cite{ashikhmin2012a}. 
	The original form of this technique used simplistic \abb{MRC}, which does not suppress interference very well, thus it required a huge number of antennas to be effective \cite{Ashikhmin2018a}.  
	The latest version of this decoding design, called \textit{large-scale fading decoding} (\abb{LSFD}) \cite{adhikary2017a}, was designed to be useful also with a practical number of antennas. 
	In the two-layer \abb{LSFD} framework, each \abb{BS} applies an arbitrary local linear decoding method in the first layer, preferably one that suppresses intra-cell interference. 
	The result is then gathered at a common central station that applies so-called \abb{LSFD} vectors in a second-layer to combine the signals from multiple \abb{BS}s to suppress pilot contamination and inter-cell interference. 
	This new decoding design overcomes the aforementioned limitations in \cite{ashikhmin2012a} and attains high \abb{SE} even with a limited number of \abb{BS} antennas. 
	
	To explain why \abb{LFSD} vectors are necessary to mitigate pilot contamination, we consider a toy example comprising of two BSs, each serving one user with the same index as their BS. There are four different channels and $\mathbf{h}_{i,j} \sim \mathcal{CN}(\mathbf{0}, \beta_{i,j} \mathbf{I}_M)$ denotes the channel between BS~$i$ and user~$j$ for  $i,j  \in \{1,2 \}$. Let $s_i$ denote the desired signal from the user in cell $i$. When using single-layer decoding with \abb{MRC}, the noise vanishes as $M \to \infty$, but pilot contamination remains \cite{marzetta2010noncooperative}. The resulting detected signals $\hat{s}_1, \hat{s}_2$ at the two BSs are then given by
	\begin{equation} \label{eq:DecodedSig2Users}
	\begin{bmatrix}
	\hat{s}_1\\ 
	\hat{s}_2
	\end{bmatrix} = 
	\begin{bmatrix}
	\beta_{1,1} s_1 +  \beta_{1,2} s_2 \\ 
	\beta_{2,1} s_1 +   \beta_{2,2} s_2
	\end{bmatrix} = 
	\underbrace{\begin{bmatrix}
		\beta_{1,1} &  \beta_{1,2}\\ 
		\beta_{2,1} &  \beta_{2,2}
		\end{bmatrix}}_{\triangleq \mathbf{B}} \, 	\begin{bmatrix}
	s_1\\ 
	s_2
	\end{bmatrix}.
	\end{equation}
	Since each BS observes a linear combination of the two signals, the asymptotic SE achieved with single-layer decoding is limited due to interference. 
	However, in a two-layer decoding system, a central station can process $\hat{s}_1$ and $\hat{s}_2$ to remove the interference as follows:
	\begin{equation} \label{eq:DecodedSig2Users-part2}
	\mathbf{B}^{-1} \begin{bmatrix}
	\hat{s}_1\\ 
	\hat{s}_2
	\end{bmatrix} =
	\mathbf{B}^{-1} \mathbf{B}
	\begin{bmatrix}
	s_1\\ 
	s_2
	\end{bmatrix}
	=  	\begin{bmatrix}
	s_1\\ 
	s_2
	\end{bmatrix}.
	\end{equation}
	The rows of the inverse matrix $\mathbf{B}^{-1}$ are called the \abb{LFSD} vectors and only depends on the statistical parameters $\beta_{1,1}, \beta_{1,2}, \beta_{2,1},\beta_{2,2}$, so the central station does not need to know the instantaneous channels.
	Since the resulting signals in \eqref{eq:DecodedSig2Users-part2} are free from noise and interference, the network can achieve an unbounded SE as $M \to \infty$.
	
	This motivating example, adapted from \cite{ashikhmin2012a}, exploits the fact that the channels are spatially uncorrelated and requires an infinite number of antennas. 
	The prior works \cite{ashikhmin2012a, adhikary2017a} are only considering uncorrelated Rayleigh fading channels and rely on the particular asymptotic properties of that channel model. The generalization to more practical correlated channels is non-trivial and has not been considered until now.\footnote{The concurrent work \cite{Adhikary2018a} appeared just as we were submitting this paper. It contains the uplink \abb{SE} for correlated Rayleigh fading described by the one-ring model and \abb{MMSE} estimation, while we consider arbitrary spatial correlation and uses two types of channel estimators. Moreover, they consider joint power control and \abb{LFSD} for max-min fairness, while we consider sum \abb{SE} optimization, making the papers complementary.} 
	In this paper, we consider spatially correlated channels with a finite number of antennas. We stress that these generalizations are practically important:  if two-layer decoding will ever be implemented in practice, the channels will be subject to spatial correlation and the BSs will have a limited number of antennas.
	
	\subsection{Main Contributions}
	%\vspace*{-0.25cm}
	In this paper, we generalize the \abb{LSFD} method from \cite{ashikhmin2012a, adhikary2017a} to a scenario with correlated Rayleigh fading and arbitrary first-layer decoders, and also develop a method for data power control in the system.  We evaluate the performance by deriving an \abb{SE} expression for the system.  Our main contributions are summarized as follows:
	\begin{itemize}
		
		\item An uplink per-user \abb{SE} is derived as a function of the second-layer decoding weights.  
		A closed-form expression is then obtained for correlated Rayleigh fading and a system that uses \abb{MRC} in the first decoding layer and an arbitrary choice of \abb{LSFD} in the second layer.   
		The  second-layer decoding weights that maximize the \abb{SE} follows in closed form.
		
		\item An uplink sum \abb{SE} optimization problem with power constraints is formulated.  Because it is non-convex and \abb{NP}-hard, we propose an alternating optimization approach that converges to a local optimum with polynomial complexity.
		
		\item Numerical results demonstrate the effectiveness of two-layer decoding for Massive \abb{MIMO} communication systems with correlated Rayleigh fading.
	\end{itemize}
	
	The rest of this paper is organized as follows: Multi-cell Massive \abb{MIMO} with two-layer decoding is presented in Section~\ref{Section: System Model}.  An \abb{SE} for the uplink together with the optimal \abb{LSFD} design is presented in Section~\ref{Section:ULPerformance}.  A maximization problem for the sum \abb{SE} is formulated and a solution is proposed in Section~\ref{Section:SumRateOpt}.  Numerical results in Section~\ref{Section:NumericalResults} demonstrate the performance of the proposed system.  Section~\ref{Section:Conclustion} states the major conclusions of the paper.
	
	\textbf{Reproducible research:} All the simulation results can be reproduced using the Matlab code and data files available at: \url{https://github.com/emilbjornson/large-scale-fading-decoding}

	\textbf{Notation}: Lower and upper case bold letters are used for vectors and matrices.  The expectation of the random variable $X$ is denoted by $\Exp \{X \}$ and the Euclidean norm of the vector $\mathbf{x}$ by $\| \mathbf{x}\|$. The transpose and Hermitian transpose of a matrix $\mathbf{M}$ are written as $\mathbf{M}^\tr$ and $\mathbf{M}^\conjtr$, respectively. The $L\times L$\mbox{-}dimensional diagonal matrix with the diagonal elements $d_1, d_2,\ldots,d_L$ is denoted $\operatorname{diag}(d_1,d_2,\ldots,d_L)$. $\mathfrak{Re}(\cdot)$ and $\mathfrak{Im}(\cdot)$ are the real and imaginary parts of a complex number. $\nabla g (\mathbf{x})|_{ \mathbf{x}_0 } $ denotes the gradient of a multivariate function $g$  at $\mathbf{x} = \mathbf{x}_0 $. Finally, $\mathcal{CN}(\boldsymbol{0}, \mathbf{R})$ is a vector of circularly symmetric, complex, jointly Gaussian distributed random variables with zero mean and correlation matrix $\mathbf{R}$. 
	
	\section{System Model} \label{Section: System Model}
	%\vspace*{-0.25cm}
	We consider a network with $L$~cells. Each cell consists of a \abb{BS} equipped with $M$ antennas that serves $K$ single-antenna users.\footnote{In the uplink, the considered network consists of multiple interfering single-input multiple-output (SIMO) channels. Such a setup has been referred to as multiuser \abb{MIMO} in the information theoretic-literature for decades, which is why we adopt this terminology in the paper.} The $M$\mbox{-}dimensional channel vector in the uplink between user~$k$ in cell~$l$ and \abb{BS}~$l'$ is denoted by $\mathbf{h}_{l,k}^{l'} \in \mathbb{C}^M$. 
	We consider the standard block-fading model, where the channels are static within a coherence block of size $\tau_\mathrm{c}$ channel uses and assume one independent realization in each block, 
	according to a stationary ergodic random process. 
	Each channel follows a correlated Rayleigh fading model:
	\begin{equation}
	\mathbf{h}_{l,k}^{l'} \sim \mathcal{CN} \left( \mathbf{0}, \mathbf{R}_{l,k}^{l'} \right),
	\end{equation}
	where $\mathbf{R}_{l,k}^{l'} \in \mathbb{C}^{M \times M}$ is the \emph{spatial correlation matrix of the channel}. 
	The \abb{BS}s know the channel statistics,
	but have no prior knowledge of the channel realizations, which need to be estimated in every coherence block.
	
	\subsection{Channel Estimation}
	%\vspace*{-0.25cm}
	As in conventional Massive \abb{MIMO} \cite{bjornson2018a},  
	the channels are estimated by letting the users transmit $\tau_\text{p}$-symbol long \emph{pilots} in a dedicated part of the coherence block, called the \emph{pilot phase}. 	
	All the cells share a common set of $\tau_\text{p}=K$ mutually orthogonal pilots $\{\pmb{\phi}_1, \ldots, \pmb{\phi}_K \}$, where the pilot $\pmb{\phi}_k \in \mathbb{C}^{\tau_p}$ spans $\tau_p$ symbols. Such orthogonal pilots are disjointly distributed among the $K$ users in each cell:
	\begin{equation}
	\pmb{\phi}_k^\conjtr \pmb{\phi}_{k'} = 
	\begin{cases}
	\tau_\text{p} & k = k',\\ 
	0 & k \neq k.
	\end{cases}
	\end{equation}
	Without loss of generality, we assume that all the users in different cells, which share the same index, use the same pilot and thereby cause pilot contamination to each other \cite{marzetta2010noncooperative}.  
	
	During the pilot phase, at BS~$l$, the signals received in the pilot phase are collectively denoted by the $M\times\tau_\text{p}$\mbox{-}dimensional matrix $\mathbf{Y}_l$ and it is given by
	\begin{equation}
	\mathbf{Y}_l = \sum_{l'=1}^L \sum_{k=1}^K \sqrt{\hat{p}_{l',k}}  \mathbf{h}_{l',k}^l \boldsymbol{\phi}_{k}^\conjtr + \mathbf{N}_l,
	\end{equation}
	where $\hat{p}_{l',k}$ is the power of the pilot of user~$k$ in cell~$l'$ and $\mathbf{N}_l$ is a matrix of independent and identically distributed noise terms, each distributed as $\mathcal{CN}(0,\sigma^2)$.
	
	An intermediate observation of the channel from user~$k$ to BS~$l$ is obtained through correlation with the pilot of user~$k$ in the following way:
	\begin{equation}
	\tilde{\mathbf{y}}_{l,k} = \mathbf{Y}_l \boldsymbol{\phi}_k = \tau_{\text{p}} \sqrt{\hat{p}_{l,k}}  \mathbf{h}_{l,k}^l +  \sum_{ \substack{l' =1 \\ l' \neq l} }^L \tau_{\text{p}} \sqrt{\hat{p}_{l',k}}  \mathbf{h}_{l',k}^l  + \tilde{\mathbf{n}}_{l,k}, \label{eq:processed-pilot}
	\end{equation}
	where $\tilde{\mathbf{n}}_{l,k} \triangleq \mathbf{N}_l \pmb{\phi}_{k} \sim \mathcal{CN}(\mathbf{0}, \tau_{\text{p}} \sigma^2 \mathbf{I}_M)$ are independent over $l$ and $k$. 
	The channel estimate and estimation error of the \abb{MMSE} estimation of $\mathbf{h}_{l,k}^l$ is presented in the following lemma.
	\begin{lemma} \label{lemma:StandardMMSE}
		If \abb{BS}~$l$ uses \abb{MMSE} estimation based on the observation in \eqref{eq:processed-pilot}, the estimate of the channel between user~$k$ in cell~$l$ and BS~$l$ is
		\begin{equation}
		\hat{\mathbf{h}}_{l,k}^l = \sqrt{\hat{p}_{l,k}} \mathbf{R}_{l,k}^l \pmb{\Psi}_{l,k}^{-1} \tilde{\mathbf{y}}_{l,k},
		\end{equation}
		where $\pmb{\Psi}_{l,k} = \Exp\{ \tilde{\mathbf{y}}_{l,k} (\tilde{\mathbf{y}}_{l,k})^{\conjtr}\}/\tau_\mathrm{p}$ is given by
		\begin{equation}
		\pmb{\Psi}_{l,k} \triangleq \sum_{l'=1}^L \tau_\mathrm{p} \hat{p}_{l',k}  \mathbf{R}_{l',k}^l +  \sigma^2 \mathbf{I}_M.
		\end{equation}
		The channel estimate is distributed as
		\begin{equation}
		\hat{\mathbf{h}}_{l,k}^l \sim \mathcal{CN} \left( \mathbf{0}, \tau_\mathrm{p} \hat{p}_{l,k} \mathbf{R}_{l,k}^l \pmb{\Psi}_{l,k}^{-1} \mathbf{R}_{l,k}^l \right),
		\end{equation}
		and the channel estimation error, $\mathbf{e}_{l,k}^l \triangleq \mathbf{h}_{l,k}^l - \hat{\mathbf{h}}_{l,k}^l$, is independently distributed as
		\begin{equation}
		\mathbf{e}_{l,k}^l \sim \mathcal{CN} \left( \mathbf{0}, \mathbf{R}_{l,k}^l - \tau_{\text{p}} \hat{p}_{l,k} \mathbf{R}_{l,k}^l \pmb{\Psi}_{l,k}^{-1} \mathbf{R}_{l,k}^l \right).
		\end{equation}
	\end{lemma}
	\begin{proof}
		This lemma follows from adopting standard \abb{MMSE} estimation results from \cite{Kay1993a}, \cite[Section~3]{Bjornson2017bo} to our system model and notation.
	\end{proof}
	Lemma~\ref{lemma:StandardMMSE} provides statistical information for the BS to construct the decoding and precoding vectors for the up- and downlink data transmission.  
	However, to compute the \abb{MMSE} estimate, the inverse matrix of $\pmb{\Psi}_{l,k}$ has to be computed for every user, which can lead to a computational complexity that might be infeasible when there are many antennas.  
	This motivates us to use the simpler estimation technique called \emph{element-wise \abb{MMSE}} (\abb{EW-MMSE}) \cite{Bjornson2017bo}. 
	
	To simplify the presentation, we make the standard assumption that the correlation matrix $\mathbf{R}_{l,k}^{l'}$ has equal diagonal elements, denoted by $\beta_{l,k}^{l'}$. This assumption is well motivated for elevated macro BSs that only observe far-field scattering effects from every cell. However, \abb{EW-MMSE} estimation of the channel can also be done when the diagonal elements are different.  
	The generalization to this case is straightforward. \abb{EW-MMSE} estimation is given in Lemma~\ref{Lemma:E-MMSE} together with the statistics of the estimates.
	\begin{lemma} \label{Lemma:E-MMSE}
		If BS~$l$ uses \abb{EW-MMSE} estimation and the diagonal elements of the  spatial correlation matrix of the channel are equal, the channel estimate between user~$k$ in cell~$l$ and BS~$l$ is
		\begin{equation}
		\hat{\mathbf{h}}_{l,k}^l = \varrho_{l,k}^l \tilde{\mathbf{y}}_{l,k},
		\end{equation}
		where 
		\begin{equation}
		\varrho_{l,k}^l = \frac{\sqrt{\hat{p}_{l,k}}  \beta_{l,k}^l}{\sum_{l'=1 }^L \tau_{\text{p}} \hat{p}_{l',k}  \beta_{l',k}^l + \sigma^2}, 
		\end{equation}
		and the channel estimate and estimation error of $\mathbf{h}_{l,k}^l$ are distributed as
		\begin{align}
		\hat{\mathbf{h}}_{l,k}^l &\sim \mathcal{CN} \left( \mathbf{0},  (\varrho_{l,k}^l)^2 \tau_{\text{p}} \pmb{\Psi}_{l,k} \right),\\
		\mathbf{e}_{l,k}^l &\sim \mathcal{CN} \left( \mathbf{0}, \mathbf{R}_{l,k}^l - (\varrho_{l,k}^l)^2 \tau_{\text{p}} \pmb{\Psi}_{l,k} \right)
		\end{align}
		and are not independent.
	\end{lemma}
	\begin{proof}
		The statistics of the estimate and estimation error follow from straightforward computation of the correlation matrices and the derivation is therefore omitted.
	\end{proof}
	
	As compared to \abb{MMSE} estimation, \abb{EW-MMSE} estimation simplifies the computations, since no inverse matrix computation is involved.
	Moreover, each \abb{BS} only needs to know the diagonal of the spatial correlation matrices, which are easier to acquire in practice than the full matrices. 
	We can also observe the relationship between two users utilizing nonorthogonal pilots by a simple expression as shown in Corollary~\ref{Corollary1}. 
	
	\begin{corollary}\label{Corollary1}
		When the diagonal elements of the spatial correlation matrix of the channel are equal, the two \abb{EW-MMSE} estimates $\hat{\mathbf{h}}_{l,k}^l$ and $\hat{\mathbf{h}}_{l'',k}^l$ of the channels of users~$k$ in cells~$l$ and $l''$ that are computed at BS~$l$ are related as:
		\begin{equation} \label{eq:Relationship2Channels}
		\frac{\hat{\mathbf{h}}_{l,k}^l}{\sqrt{\hat{p}_{l,k}} \beta_{l,k}^l}  = \frac{\hat{\mathbf{h}}_{l'',k}^l}{\sqrt{\hat{p}_{l'',k}} \beta_{l'',k}^l},
		\end{equation}
		where $\hat{\mathbf{h}}_{l'',k}^l = \varrho_{l'',k}^l \tilde{\mathbf{y}}_{l,k}$ with
		\begin{equation}
		\varrho_{l'',k}^l = \sqrt{\hat{p}_{l'',k}}  \beta_{l'',k}^l / \left(\sum_{l'=1 }^L \tau_{\text{p}} \hat{p}_{l',k}  \beta_{l',k}^l + \sigma^2 \right).
		\end{equation}
	\end{corollary}
	Corollary~\ref{Corollary1} mathematically shows that the channel estimates of two users with the same pilot signal only differ from each other by a scaling factor. 
	Using \abb{EW-MMSE} estimation leads to severe pilot contamination that cannot be mitigated by linear processing of the data signal only, at least not with the approach in \cite{bjornson2018a}.
	
	\begin{figure}[t]
		\centering
		\includegraphics[trim=0.0cm 0cm 0cm 0cm, clip=true, width=2.7in]{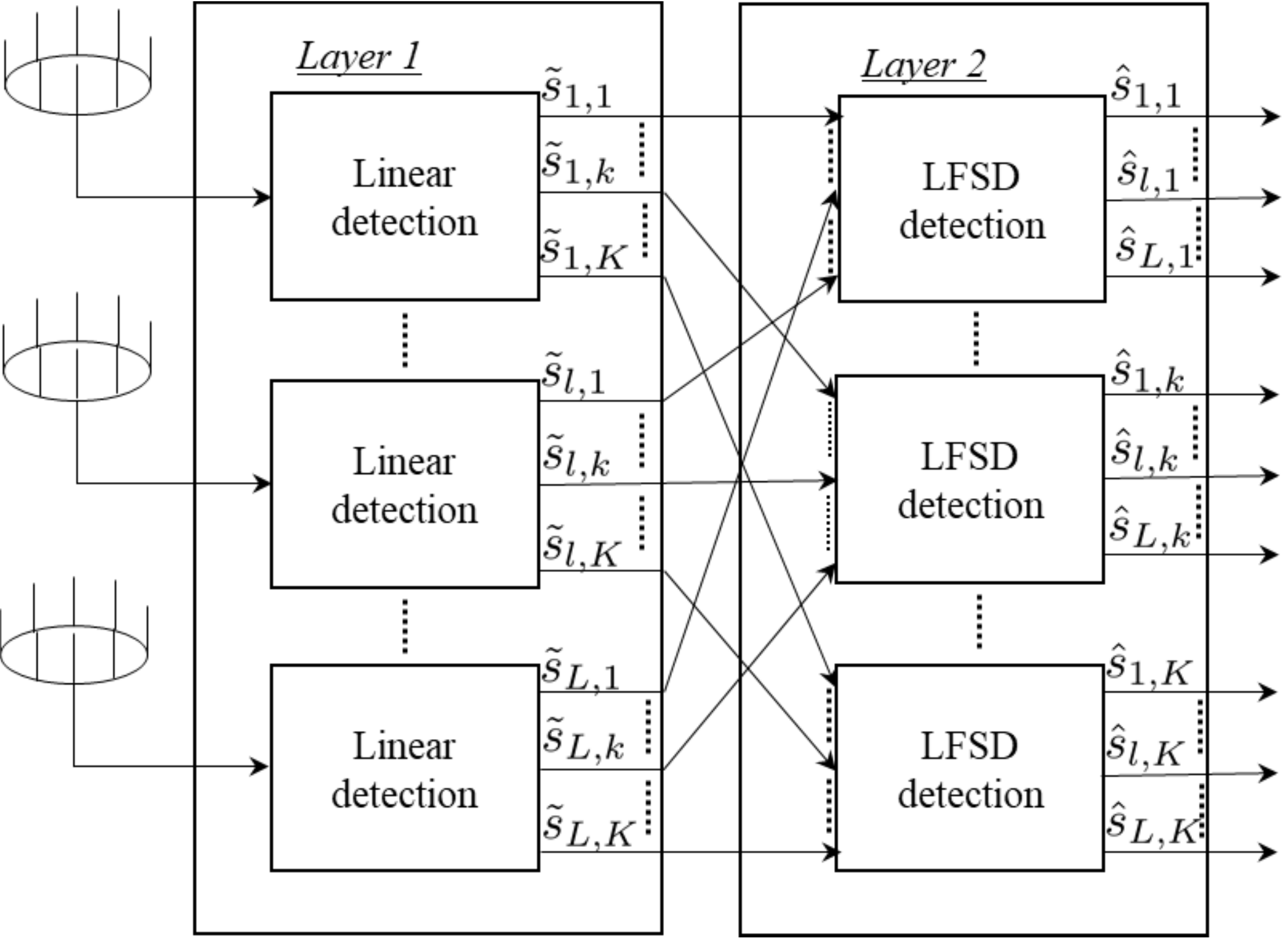} %\vspace*{-1.2cm}
		\caption{Desired signals are detected by the two-layer decoding technique.}
		\label{fig:decoder}
	\end{figure}
	%\vspace*{-0.5cm}
	\subsection{Uplink Data Transmission}
	%\vspace*{-0.25cm}
	During the \emph{data phase}, it is assumed that user~$k$ in cell~$l'$ sends a zero-mean symbol~$s_{l',k}$ with variance $\Exp \{ |s_{l',k}|^2 \} = 1$. 
	The received signal at BS~$l$ is then
	\begin{equation} \label{eq:ReceivedSigBSl}
	\mathbf{y}_l = \sum_{l'=1}^L \sum_{k=1}^K \sqrt{p_{l',k}} \mathbf{h}_{l',k}^l s_{l',k} + \mathbf{n}_l,
	\end{equation}
	where $p_{l',k}$ denotes the transmit power of user~$k$ in cell~$l'$. 
	Based on the signals in \eqref{eq:ReceivedSigBSl}, the BSs decode the symbols with the two-layers-decoding technique that is illustrated in Fig.~\ref{fig:decoder}. The general idea of a two-layer decoding system is that each BS decodes the desired signals from its coverage area in the first layer. A central station is then collecting the decoded signals of all users that used the same pilot and jointly processes these signals in the second layer to suppress inter-cell interference using LSFD vectors. In detail, an estimate of the symbol from user~$k$ in cell~$l$ is obtained by local linear decoding in the first layer as
	\begin{equation}
	\tilde{s}_{l,k} = \mathbf{v}_{l,k}^\conjtr \mathbf{y}_l = \sum_{l'=1}^L \sum_{k'=1}^K \sqrt{p_{l',k'}} \mathbf{v}_{l,k}^\conjtr \mathbf{h}_{l',k'}^l s_{l',k'} + \mathbf{v}_{l,k}^\conjtr \mathbf{n}_l,
	\end{equation} 
	where $\mathbf{v}_{l,k}$ is the \emph{linear decoding vector}. The symbol estimate $\tilde{s}_{l,k}$ generally contains interference and, in Massive \abb{MIMO}, the pilot contamination from all the users with the same pilot sequence is particularly large. 
	To mitigate the pilot contamination, all the symbol estimates of the contaminating users are collected in a vector 
	\begin{equation}
	\tilde{\mathbf{s}}_{k} \triangleq [\tilde{s}_{1,k}, \tilde{s}_{2,k}, \ldots,  \tilde{s}_{L,k} ]^\tr \in \mathbb{C}^L.
	\end{equation}
	After the local decoding, a second layer of centralized decoding is performed on this vector using the \emph{\abb{LSFD} vector} $\mathbf{a}_{l,k} \triangleq [a_{l,k}^1, a_{l,k}^2, \ldots, a_{l,k}^L]^\tr \in \mathbb{C}^L$, where $a_{l,k}^{l'}$ is the \emph{\abb{LSFD} weight}. The final estimate of the data symbol from user~$k$ in cell~$l$ is then given by
	\begin{equation} \label{eq:919838377}
	\begin{split}
	& \hat{s}_{l,k} = \mathbf{a}_{l,k}^\conjtr \tilde{\mathbf{s}}_{k} = \sum_{l'=1}^L (a_{l,k}^{l'})^{\ast} \tilde{s}_{l',k}.
	\end{split}
	\end{equation}
	
	In the next section, we use the decoded signals $\hat{s}_{l,k}$ together with the asymptotic channel properties \cite[Section~2.5]{Bjornson2017bo} to derive a closed-from expression of an uplink \abb{SE}.
	
	\section{Uplink Performance Analysis} \label{Section:ULPerformance}
	This section first derives a general \abb{SE} expression for each user~$k$ in each cell~$l$ and a closed-form expression when using \abb{MRC}. 
	These expressions are then used to obtain the \abb{LSFD} vectors that maximize the \abb{SE}. 
	The use-and-then-forget capacity bounding technique \cite[Chapter~2.3.4]{Marzetta2016a}, \cite[Section~4.3]{bjornson2018a} allows us to compute a lower bound on the uplink ergodic capacity (i.e., an achievable \abb{SE}). We first rewrite \eqref{eq:919838377} as
	\begin{align} \label{eq:UseandForget}
	&\hat{s}_{l,k} = \vphantom{\sum_{\substack{k'=1 \\ k'\neq k}}^K} \sum_{l'=1}^L (a_{l,k}^{l'})^{\ast} \Exp\{\mathbf{v}_{l',k}^\conjtr \mathbf{h}_{l,k}^{l'} \} \sqrt{p_{l,k}}s_{l,k}\notag\\
	&+ \sum_{l'=1}^L (a_{l,k}^{l'})^{\ast} \sum_{\substack{l''=1 \\ l'' \neq l }}^L \Exp \{\mathbf{v}_{l',k}^\conjtr \mathbf{h}_{l'',k}^{l'} \} \sqrt{p_{l'',k}}s_{l'',k}\notag\\
	&+\sum_{l'=1}^L (a_{l,k}^{l'})^{\ast} \sum_{l''=1}^L \left(\mathbf{v}_{l',k}^\conjtr \mathbf{h}_{l'',k}^{l'} - \Exp \{\mathbf{v}_{l',k}^\conjtr \mathbf{h}_{l'',k}^{l'} \} \right) \sqrt{p_{l'',k}} s_{l'',k}\notag\\
	&+\sum_{l'=1}^L (a_{l,k}^{l'})^{\ast} \sum_{l''=1}^L \sum_{\substack{k'=1 \\ k'\neq k}}^K \sqrt{p_{l'',k'}} \mathbf{v}_{l',k}^\conjtr \mathbf{h}_{l'',k'}^{l'} s_{l'',k'}\notag\\
	&+\vphantom{\sum_{\substack{k'=1 \\ k'\neq k}}^K}  \sum_{l'=1}^L (a_{l,k}^{l'})^{\ast} \mathbf{v}_{l',k}^\conjtr \mathbf{n}_{l'},
	\end{align}
	then by considering the first part of \eqref{eq:UseandForget} as the desired signal from user~$k$ in cell~$l$ while the remaining is effective Gaussian noise, a lower bound on the uplink ergodic capacity is shown in Lemma~\ref{lemma:General_Rate}.
	
	\begin{lemma} \label{lemma:General_Rate}
		A lower bound on the uplink ergodic capacity is
		\begin{equation} \label{eq:GeneralSE}
		R_{l,k} = \max_{\{a_{l,k}^{l'}\}} \left( 1- \frac{\tau_\mathrm{p}}{\tau_\mathrm{c}} \right) \log_2 \left(1 + \mathrm{SINR}_{l,k} \right),
		\end{equation}
		where the effective SINR, denoted by $\mathrm{SINR}_{l,k}$, is 
		\begin{equation} \label{eq:GeneralSINR}
		\mathrm{SINR}_{l,k} = \Exp\{ |\mathtt{DS}_{l,k}|^2 \} / D_{l,k},
		\end{equation}
		where $D_{l,k}$ is given by
		\begin{multline}
		D_{l,k} = \Exp \{ |\mathtt{PC}_{l,k} |^2 \} + \Exp \{ |\mathtt{BU}_{l,k}|^2\}\\
		+\sum\limits_{l'=1}^{L} \sum\limits_{\substack{k'=1\\k'\neq k}}^K \Exp \{ |\mathtt{NI}_{l',k'}|^2\} + \Exp \{ |\mathtt{AN}_{l,k}|^2\}.
		\end{multline}
		Here $\mathtt{DS}_{l,k}, \mathtt{PC}_{l,k}, \mathtt{BU}_{l,k},\mathtt{NI}_{l',k'},$ and $\mathtt{AN}_{l,k}$ stand for the desired signal, the pilot contamination, the beamforming gain uncertainty, the non-coherent interference, and the additive noise, respectively, whose expectations are defined as
		\begin{align}
		\mathbb{E} \{ |\mathtt{DS}_{l,k }|^2 \} &\triangleq p_{l,k} \left| \sum_{l'=1}^L (a_{l,k}^{l'})^{\ast} \Exp \{\mathbf{v}_{l',k}^\conjtr \mathbf{h}_{l,k}^{l'} \}\right|^2 ,\\
		\mathbb{E} \{ |\mathtt{PC}_{l,k}|^2 \} &\triangleq  \sum_{\substack{l''=1 \\ l'' \neq l }}^L p_{l'',k} \left| \sum_{l'=1}^L (a_{l,k}^{l'})^{\ast}  \Exp \{\mathbf{v}_{l',k}^\conjtr \mathbf{h}_{l'',k}^{l'} \} \right|^2,%\\
		\end{align}
		\begin{align}
		\mathbb{E} \{|\mathtt{BU}_{l,k}|^2 \} &\triangleq \sum_{l'=1}^L p_{l',k} \Exp \Bigg\{ \Bigg| \sum_{l''=1}^L (a_{l,k}^{l''})^{\ast} \Bigg( \mathbf{v}_{l'',k}^\conjtr \mathbf{h}_{l',k}^{l''} -  \notag \\ & \quad \Exp \{\mathbf{v}_{l'',k}^\conjtr \mathbf{h}_{l',k}^{l''} \} \Bigg) \Bigg|^2 \Bigg\},\\
		\mathbb{E} \{ |\mathtt{NI}_{l',k'}|^2 \} &\triangleq  p_{l',k'} \Exp \left\{ \left| \sum_{l''=1}^L (a_{l,k}^{l''})^{\ast}  \mathbf{v}_{l'',k}^\conjtr \mathbf{h}_{l',k'}^{l''} \right|^2 \right\} , \\
		%\end{align}
		%\begin{align}
		\mathbb{E} \{ |\mathtt{AN}_{l,k}|^2 \} &\triangleq \Exp \left\{ \left| \sum_{l'=1}^L (a_{l,k}^{l'})^\ast  (\hat{\mathbf{h}}_{l',k}^{l'})^\conjtr \mathbf{n}_{l'} \right|^2 \right\}.
		\end{align}  
	\end{lemma}
	Note that the lower bound on the uplink ergodic capacity in Lemma~\ref{lemma:General_Rate} can be applied to any linear decoding method and any \abb{LSFD} design.
	
	To maximize the \abb{SE} of user~$k$ in cell~$l$ is equivalent to selecting the \abb{LSFD} vector that maximizes a \textit{Rayleigh quotient} as shown in the proof of the following theorem. 
	This is the first main contribution of this paper.
	\begin{theorem} \label{Theorem1v1}
		For a given set of pilot and data power coefficients, the \abb{SE} of user~$k$ in cell~$l$ is
		\begin{equation} \label{eq:OptimalRate}
		R_{l,k} = \left( 1- \frac{\tau_\mathrm{p}}{\tau_\mathrm{c}} \right) \log_2 \left(1 +  p_{l,k} \mathbf{b}_{l,k}^\conjtr \left( \sum_{i=1}^4 \mathbf{C}_{l,k}^{(i)}  \right)^{-1} \mathbf{b}_{l,k} \right),
		\end{equation}
		where the matrices $\mathbf{C}_{l,k}^{(1)},\mathbf{C}_{l,k}^{(2)}, \mathbf{C}_{l,k}^{(3)}, \mathbf{C}_{l,k}^{(4)} \in \mathbb{C}^{L \times L}$ and the vector $\mathbf{b}_{l,k}$ are defined as %\comment{Chris: why double prime $l''$? Is this a typo here and in the rest of the Theorem? \\ Chien: Thanks. I have replaced $l''$ with $l'$.}
		\begin{align}
		\mathbf{C}_{l,k}^{(1)} &\triangleq \sum_{\substack{l'=1 \\ l' \neq l }}^L p_{l',k} \mathbf{b}_{l',k} \mathbf{b}_{l',k}^\conjtr, \\
		\mathbf{C}_{l,k}^{(2)} &\triangleq \sum_{l'=1}^L p_{l',k} \Exp \left\{ \tilde{\mathbf{b}}_{l',k}  \tilde{\mathbf{b}}_{l',k}^\conjtr \right\},
		\\
		\mathbf{C}_{l,k}^{(3)} &\triangleq \mathrm{diag} \left(\sum_{l'=1}^L \sum_{\substack{k'=1 \\ k' \neq k}}^K  p_{l',k'} \Exp \left\{\left|  \mathbf{v}_{1,k}^\conjtr \mathbf{h}_{l',k'}^{1} \right|^2  \right\}, \ldots, \right. \notag  \\ & \qquad \left. \sum_{l'=1}^L \sum_{\substack{k'=1 \\ k' \neq k}}^K  p_{l',k'} \Exp \left\{\left|  \mathbf{v}_{L,k}^\conjtr \mathbf{h}_{l',k'}^{L} \right|^2  \right\} \right),\\
		\mathbf{C}_{l,k}^{(4)} &\triangleq \mathrm{diag} \left( \sigma^2 \Exp \left\{ \| \mathbf{v}_{1,k} \|^2 \right\}, \ldots, \sigma^2  \Exp  \left\{ \| \mathbf{v}_{L,k} \|^2 \right\} \right),
		\end{align}
		and the vectors $\mathbf{b}_{l',k}, \tilde{\mathbf{b}}_{l',k} \in \mathbb{C}^L, \forall l' = 1,\ldots,L,$ are defined as
		\begin{align}
		\mathbf{b}_{l',k} &\triangleq \left[ 
		\Exp  \{ \mathbf{v}_{1,k}^H \mathbf{h}_{l',k}^1 \}, \ldots, \Exp  \{ \mathbf{v}_{L,k}^H \mathbf{h}_{l',k}^L \} \right]^\tr, \label{eq:bikv1}\\
		\tilde{\mathbf{b}}_{l',k} &\triangleq \left[ \mathbf{v}_{1,k}^\conjtr \mathbf{h}_{l',k}^{1} ,  \ldots, \mathbf{v}_{L,k}^\conjtr \mathbf{h}_{l',k}^{L} \right]^\tr - \mathbf{b}_{l',k}.
		\end{align}
		In order to attain this \abb{SE}, the \abb{LSFD} vector is formulated as
		\begin{equation} \label{eq:LSFDVector}
		\mathbf{a}_{l,k} = \left( \sum_{i=1}^4 \mathbf{C}_{l,k}^{(i)} \right)^{-1} \mathbf{b}_{l,k}, \quad\forall l,k.
		\end{equation}
	\end{theorem}
	\begin{proof}
		The proof is available in Appendix~\ref{Appendix: Proof-Theorem-v1}.
	\end{proof}
	We stress that the \abb{LSFD} vector in \eqref{eq:LSFDVector} is designed to maximize the \abb{SE} in \eqref{eq:OptimalRate} for every user in the network for a given data and pilot power and a given first-layer decoder. Note that Theorem~\ref{Theorem1v1} can be applied to practical correlated Rayleigh fading channels with either \abb{MMSE} or \abb{EW-MMSE} estimation and any conceivable choice of first-layer decoder. This stands in contrast to the previous work \cite{nayebi2016,adhikary2017a} that only considered uncorrelated Rayleigh fading channels, which are unlikely to occur in practice, and particular linear combining methods that were selected to obtained closed-form expressions.
	Theorem~\ref{Theorem1v1} explicitly reveals the influence that mutual interference and noise have on the \abb{SE} when utilizing the optimal \abb{LFSD} vector given in \eqref{eq:LSFDVector}:  $\mathbf{C}_{l,k}^{(1)}$ determines the amount of remaining pilot contamination from the $(L-1)$ users using the same pilot sequence as user~$k$ in cell~$l$. The beamforming gain uncertainty is represented by $\mathbf{C}_{l,k}^{(2)}$, while $\mathbf{C}_{l,k}^{(3)}$  is the noncoherent mutual interference from the remaining users and $\mathbf{C}_{l,k}^{(4)}$ represent the additive noise.
	
	The following theorem states a closed-form expression of the \abb{SE} for the case of \abb{MRC}, i.e., $\mathbf{v}_{l,k} = \hat{\mathbf{h}}_{l,k}^l$. 
	This is the second main contribution of this paper.
	\begin{theorem} \label{Theorem1}
		When \abb{MRC} is used, the \abb{SE} in \eqref{eq:GeneralSE} of user~$k$ in cell~$l$ is given by
		\begin{equation} \label{eq: ClosedForm_Rate_MMSE}
		R_{l,k} = \left( 1- \frac{\tau_{\text{p}}}{\tau_c} \right) \log_2 \left(1 + \mathrm{SINR}_{l,k} \right),
		\end{equation}
		where the SINR value is given in \eqref{eq:MMSE_SINR}.
		\begin{figure*}
			\begin{equation} \label{eq:MMSE_SINR}
			\mathrm{SINR}_{l,k} = \frac{  p_{l,k} \left| \sum_{l'=1}^L (a_{l,k}^{l'})^{\ast} b_{l,k}^{l'} \right|^2  }{  \sum_{\substack{l'=1 \\ l' \neq l }}^L p_{l',k} \left| \sum_{l''=1}^L (a_{l,k}^{l''})^{\ast} b_{l',k}^{l''} \right|^2 + \sum_{l'=1}^L \sum_{k'=1}^K  \sum_{l''=1}^L  p_{l',k'} |a_{l,k}^{l''}|^2 c_{l'',k}^{l',k'}  + \sum_{l'=1}^L |a_{l,k}^{l'}|^2 d_{l',k} }.
			\end{equation} \hrule
		\end{figure*}
		The values $b_{l',k}^{l''}, c_{l'',k}^{l',k'},$ and $d_{l',k}$ are different depending on the channel estimation technique. \abb{MMSE} estimation results in
		\begin{align}
		b_{l',k}^{l''} &= \sqrt{\tau_\mathrm{p} \hat{p}_{l',k} \hat{p}_{l'',k}} \mathrm{tr} \left( \pmb{\Psi}_{l'',k}^{-1} \mathbf{R}_{l'',k}^{l''} \mathbf{R}_{l',k}^{l''} \right),\label{eq:8926616}\\
		c_{l'',k}^{l',k'} &= \hat{p}_{l'',k} \mathrm{tr} \left( \mathbf{R}_{l'',k}^{l''} \pmb{\Psi}_{l'',k}^{-1} \mathbf{R}_{l'',k}^{l''} \mathbf{R}_{l',k'}^{l''} \right),\\
		d_{l',k} &= \sigma^2 \hat{p}_{l',k} \mathrm{tr} \left( \pmb{\Psi}_{l',k}^{-1} \mathbf{R}_{l',k}^{l'} \mathbf{R}_{l',k}^{l'} \right),
		\end{align}
		whereas \abb{EW-MMSE} results in
		\begin{align}
		b_{l',k}^{l''} &= \sqrt{\tau_\mathrm{p}} \varrho_{l'',k}^{l''} \varrho_{l',k}^{l''} \mathrm{tr} \left( \pmb{\Psi}_{l'',k} \right),\\
		c_{l'',k}^{l',k'} &=  ( \varrho_{l'',k}^{l''})^2 \mathrm{tr} \left( \mathbf{R}_{l',k'}^{l''} \pmb{\Psi}_{l'',k} \right),\\
		d_{l',k} &= ( \varrho_{l',k}^{l'})^2 \sigma^2 \mathrm{tr} \left( \pmb{\Psi}_{l',k} \right).
		\end{align}
	\end{theorem}
	\begin{proof}
		The proofs consist of computing the moments of complex Gaussian distributions.  They are available in Appendix~\ref{Appendix: Proof-Theorem-1} and Appendix~\ref{Appendix: Proof-Theorem-2} for \abb{MMSE} and \abb{EW-MMSE} estimation, respectively.
	\end{proof}
	
	Theorem~\ref{Theorem1} describes the exact impact of the spatial correlation of the channel on the system performance through the coefficients $b_{l',k}^{l''}, c_{l'',k}^{l',k'},$ and $d_{l',k}$.  
	It is seen that the numerator of \eqref{eq:MMSE_SINR} grows as the square of the number of antennas, $M^2$, since the trace in \eqref{eq:8926616} is the sum of $M$ terms. This gain comes from the coherent combination of the signals from the $M$ antennas.  It can also be seen from Theorem~\ref{Theorem1} that the pilot contamination in \eqref{eq:919838377} combines coherently, i.e., its variance---the first term in the denominator that contains $b^{l''}_{l,k}$---grows as $M^2$.  The other terms in the denominator represent the impact of non-coherent interference and Gaussian noise, respectively.  These two terms only grow as $M$. Since the interference terms contain products of correlation matrices of different users, the interference is smaller between users that have very different spatial correlation characteristics \cite{Bjornson2017bo}.
	
	The following corollary gives the optimal \abb{LSFD} vector $\mathbf{a}_{l,k}$ that maximizes the \abb{SE} of every user in the network for a given set of pilot and data powers, 
	which is expected to work well when each \abb{BS} is equipped with a practical number of antennas. 
	\begin{corollary} \label{corollary:Opt_LSFD}
		For a given set of data and pilot powers, by using \abb{MRC} and \abb{LSFD}, the \abb{SE} in Theorem~\ref{Theorem1} is given in the closed form as
		\begin{equation} \label{eq:MaxRateMRC}
		R_{l,k} = \left( 1- \frac{\tau_\mathrm{p}}{\tau_\mathrm{c}} \right) \log_2 \left(1 +  p_{l,k}\mathbf{b}_{l,k}^\conjtr  \mathbf{C}_{l,k}^{-1} \mathbf{b}_{l,k} \right)
		\end{equation}
		where $\mathbf{C}_{l,k}\in \mathbb{C}^{L \times L}$ and $\mathbf{b}_{l,k} \in \mathbb{C}^{L}$ are defined as
		\begin{align}
		\mathbf{C}_{l,k} &\triangleq \sum_{ \substack{l'=1\\l' \neq l}}^L  p_{l',k} \mathbf{b}_{l',k} \mathbf{b}_{l',k}^\conjtr + \mathrm{diag} \left( \sum_{l'=1}^L \sum_{k'=1}^K p_{l',k'} c_{1,k}^{l',k'} + d_{1,k}, \right. \notag \\ 
		& \quad \left. \ldots, \sum_{l'=1}^L \sum_{k'=1}^K p_{l',k'} c_{L,k}^{l',k'} + d_{L,k} \right), \label{eq:Clk}\\
		\mathbf{b}_{l',k} &\triangleq [b_{l',k}^1, \ldots, b_{l',k}^L]^\tr. \label{eq:bik}
		\end{align}
		The \abb{SE} in \eqref{eq:MaxRateMRC} is obtained by using \abb{LSFD} vector defined as
		\begin{equation}
		\mathbf{a}_{l,k} = \mathbf{C}_{l,k}^{-1} \mathbf{b}_{l,k}.
		\end{equation}
	\end{corollary}
	Even though Corollary~\ref{corollary:Opt_LSFD} is a special case of Theorem~\ref{Theorem1v1} when \abb{MRC} is used, its contributions  are two-fold: The \abb{LSFD} vector $\mathbf{a}_{l,k}$ is  computed in the closed form which is independent of the small-scale fading, so it is easy to compute and store. Moreover, this \abb{LSFD} vector is the generalization of the vector given in \cite{adhikary2017a} to the larger class of correlated Rayleigh fading channels.
	
	\section{Data Power Control and \abb{LFSD} Design for Sum \abb{SE} Optimization} \label{Section:SumRateOpt}
	In this section, how to choose the powers $\{p_{l,k}\}$ (power control) and the \abb{LSFD} vector to maximize the sum \abb{SE} is investigated.  
	The sum \abb{SE} maximization problem for a multi-cell Massive \abb{MIMO} system is first formulated based on results from previous sections.  
	Next, an iterative algorithm based on solving a series of convex optimization problems is proposed to efficiently find a stationary point.
	
	\subsection{Problem Formulation}
	We consider sum \abb{SE} maximization:
	\begin{equation} \label{Problem: Sumrate}
	\begin{aligned} 
	& \underset{ \{ p_{l,k} \geq 0 \}, \{ \mathbf{a}_{l,k} \} }{\textrm{maximize}}
	&&  \sum_{l=1}^L \sum_{k=1}^K R_{l,k} \\
	& \mathrm{subject\,\,to}
	& & p_{l,k} \leq P_{\mathrm{max},l,k} \quad \forall l,k.
	\end{aligned}
	\end{equation}
	Using the rate \eqref{eq: ClosedForm_Rate_MMSE} in \eqref{Problem: Sumrate}, and removing the constant pre-log factor, we obtain the equivalent formulation
	\begin{equation} \label{Problem: Sumratev1}
	\begin{aligned}
	& \underset{ \{ p_{l,k} \geq 0 \},  \{ \mathbf{a}_{l,k} \}  }{\textrm{maximize}}
	&&  \sum_{l=1}^L \sum_{k=1}^K \log_2 \left(1 + \mathrm{SINR}_{l,k} \right) \\
	& \mathrm{subject\,\,to}
	& & p_{l,k} \leq P_{\mathrm{max},l,k} \quad \forall l,k.
	\end{aligned}
	\end{equation}
	This can be shown to be a non-convex and \abb{NP}-hard problem using the same methodology as in \cite{Annapureddy2011a}, even if the fine details will be different since that paper considers small-scale multi-user \abb{MIMO} systems with perfect channel knowledge.
	Therefore, the global optimum is difficult to find in general. Nevertheless, solving the ergodic sum SE maximization~\eqref{Problem: Sumratev1} for a Massive \abb{MIMO} system is more practical than maximizing the instantaneous \abb{SE}s for a small-scale \abb{MIMO} network and a given realization of the small-scale fading \cite{Shi2011, Weeraddana2012a}. 
	In contrast, the sum \abb{SE} maximization in \eqref{Problem: Sumratev1} only depends on the large-scale fading coefficients, which simplifies matters and allows the solution to be used for a long period of time. 
	Another key difference from prior work is that we jointly optimize the data powers and \abb{LSFD} vectors.
	
	Instead of seeking the global optimum to~\eqref{Problem: Sumratev1}, which has an exponential computational complexity, we will use the weighted \abb{MMSE} method \cite{Christensen2008a, brandt2014a} to obtain a stationary point to \eqref{Problem: Sumratev1} in polynomial time. This is a standard method to break down a sum SE problem into subproblems that can be solved sequentially. We stress that the resulting subproblems and algorithms are different for every problem that the method is applied to, thus our solution is a novel contribution.
	To this end, we first formulate the weighted \abb{MMSE} problem from \eqref{Problem: Sumratev1} as shown in Theorem~\ref{Theorem:MMSEOptProblem}.
	\begin{theorem} \label{Theorem:MMSEOptProblem}
		The optimization problem
		\begin{equation} \label{Problem: Sumratev2}
		\begin{aligned}
		& \underset{ \substack{\{ p_{l,k} \geq 0 \}, \{ \mathbf{a}_{l,k} \},\\ \{ w_{l,k} \geq 0 \}, \{ u_{l,k} \} }}{\mathrm{minimize}}
		&&  \sum_{l=1}^L \sum_{k=1}^K w_{l,k} e_{l,k} - \ln (w_{l,k}) \\
		& \mathrm{subject\,\,to}
		& & p_{l,k} \leq P_{\mathrm{max},l,k} \;, \forall l,k,\\
		\end{aligned}
		\end{equation}
		where $e_{l,k}$ is defined as
		\begin{multline} \label{eq: errorOpt}
		e_{l,k} \triangleq  |u_{l,k}|^2 \left( \sum\limits_{l'=1 }^L p_{l',k} \left| \sum\limits_{l''=1}^L  (a_{l,k}^{l''})^{\ast} b_{l',k}^{l''} \right|^2  \right.\\
		\left.+  \sum\limits_{l'=1}^L \sum\limits_{k'=1}^K \sum\limits_{l''=1}^L p_{l',k'} |a_{l,k}^{l''} |^2  c_{l'',k}^{l',k'} + \sum\limits_{l'=1}^L |a_{l,k}^{l'}|^2 d_{l',k}  \right)\\
		- 2\sqrt{p_{l,k}} \mathfrak{Re} \left(u_{l,k} \left( \sum_{l'=1}^L (a_{l,k}^{l'})^{\ast} b_{l,k}^{l'} \right)\right)  +1,
		\end{multline}
		is equivalent to the sum \abb{SE} optimization problem \eqref{Problem: Sumratev1} in the sense that \eqref{Problem: Sumratev1} and \eqref{Problem: Sumratev2} have the same global optimal power solution $\{p_{l,k}\}, \forall l,k,$ and the same \abb{LSFD} elements $ \{ a^{l'}_{l,k} \}, \forall l,k,l'$. 
	\end{theorem}
	\begin{proof} 
		The proof is available in Appendix~\ref{Appendix: Proof-Theorem-MMSEOptProblem}.
	\end{proof}
	
	\subsection{Iterative Algorithm}
	We now find a stationary point to  \eqref{Problem: Sumratev2} by decomposing it into a sequence of subproblems, each having a closed-form solution.  By changing variable as $\rho_{l,k} = \sqrt{p_{l,k}}$, the optimization problem \eqref{Problem: Sumratev2} is equivalent to
	\begin{equation} \label{Problem: Sumratev3}
	\begin{aligned}
	& \underset{ \substack{\{ \rho_{l,k} \geq 0 \}, \{ \mathbf{a}_{l,k} \},\\ \{ w_{l,k} \geq 0 \}, \{ u_{l,k} \} }}{\mathrm{minimize}}
	&&  \sum_{l=1}^L \sum_{k=1}^K w_{l,k} e_{l,k} - \ln (w_{l,k}) \\
	& \mathrm{subject\,\,to}
	& & \rho_{l,k}^2 \leq P_{\mathrm{max},l,k} \;, \forall l,k,\\
	\end{aligned}
	\end{equation}
	where $e_{l,k}$ is
	\begin{multline} \label{eq: errorOptv1}
	e_{l,k} =  |u_{l,k}|^2 \left( \sum\limits_{l'=1 }^L \rho_{l',k}^2 \left| \sum\limits_{l''=1}^L  (a_{l,k}^{l''})^{\ast} b_{l',k}^{l''} \right|^2  \right.\\\left.
	+ \sum\limits_{l'=1}^L \sum\limits_{k'=1}^K \sum\limits_{l''=1}^L \rho_{l',k'}^2  |a_{l,k}^{l''}|^2 c_{l'',k}^{l',k'} + \sum\limits_{l'=1}^L | a_{l,k}^{l'}|^2 d_{l',k}  \right) \\
	- 2 \rho_{l,k} \mathfrak{Re} \left( u_{l,k} \left( \sum_{l'=1}^L (a_{l,k}^{l'})^{\ast} b_{l,k}^{l'}\right) \right) + 1.
	\end{multline}
	
	As a third main contribution of this paper, the following theorem provides an algorithm that relies on alternating optimization to find a stationary point to \eqref{Problem: Sumratev3}.
	\begin{theorem} \label{Theorem: IterativeSol}
		A stationary point to \eqref{Problem: Sumratev3} is obtained by iteratively updating $\{ \mathbf{a}_{l,k}, u_{l,k}, w_{l,k}, \rho_{l,k} \}$. Let $ \mathbf{a}_{l,k}^{n-1}, u_{l,k}^{n-1}, w_{l,k}^{n-1}, \rho_{l,k}^{n-1}$ the values after iteration $n-1$. At iteration $n$, the optimization parameters are updated in the following way:
		\begin{itemize}[leftmargin=*]
			\item $u_{l,k}$ is updated as
			\begin{equation} \label{eq:u_lk_n1}
			u_{l,k}^{(n)} = \frac{\rho_{l,k}^{(n-1)}  \sum\limits_{l'=1}^L a_{l,k}^{l',(n-1)} (b_{l,k}^{l'})^{\ast} }{ \tilde{u}_{l,k}^{(n-1)}},
			\end{equation}
			where the value $\tilde{u}_{l,k}^{(n-1)}$ is defined in \eqref{eq:tildeu_lk_n1}.
			\begin{figure*}
				\begin{equation} \label{eq:tildeu_lk_n1}
				\tilde{u}_{l,k}^{(n-1)} = \sum\limits_{l'=1}^L (\rho_{l',k}^{(n-1)})^2 \left| \sum\limits_{l''=1}^L (a_{l,k}^{l'', (n-1)})^{\ast} b_{l',k}^{l''} \right|^2  + \sum\limits_{l'=1}^L \sum\limits_{k'=1}^K \sum\limits_{l''=1}^L (\rho_{l',k'}^{(n-1)})^2 | a_{l,k}^{l'',(n-1)}|^2 c_{l'',k}^{l',k'}
				+  \sum\limits_{l'=1}^L | a_{l,k}^{l',(n-1)}|^2 d_{l',k}.
				\end{equation} %\vspace*{-0.5cm}
			\end{figure*}
			
			\item $w_{l,k}$ is updated as
			\begin{equation} \label{eq:w_lk_n1}
			w_{l,k}^{(n)} = \left(e_{l,k}^{(n)}\right)^{-1},
			\end{equation}
			where $e_{l,k}^{(n)}$ is defined as
			\begin{multline} \label{eq:e_lk_n1}
			e_{l,k}^{(n)} = |u_{l,k}^{(n)}|^2 \tilde{u}_{l,k}^{(n-1)}\\ -2\rho_{l,k}^{(n-1)} \mathfrak{Re} \left(u_{l,k}^{(n)}   \left( \sum_{l'=1}^L (a_{l,k}^{l',(n-1)})^{\ast} b_{l,k}^{l'} \right) \right) + 1.
			\end{multline}
			\item $\mathbf{a}_{l,k}$ is updated as
			\begin{equation} \label{eq: a_lk_n1}
			\mathbf{a}_{l,k}^{(n)} = \frac{\tilde{u}_{l,k}^{\ast, (n)}}{\sum\limits_{l'=1}^L (a_{l,k}^{l',(n-1)})^{\ast} b_{l,k}^{l'}} \left(\widetilde{\mathbf{C}}_{l,k}^{(n-1)}\right)^{-1} \mathbf{b}_{l,k},
			\end{equation}
			where $\widetilde{\mathbf{C}}_{l,k}^{(n-1)}$ is computed as in \eqref{eq:C_lk_n}.
			\begin{figure*}
				\begin{equation} \label{eq:C_lk_n}
				\widetilde{\mathbf{C}}_{l,k}^{(n-1)} = \sum_{l'=1}^L  (\rho_{l',k}^{(n-1)})^2 \mathbf{b}_{l',k} \mathbf{b}_{l',k}^\conjtr + \mathrm{diag} \left( \sum_{l'=1}^L \sum_{k'=1}^K (\rho_{l',k'}^{(n-1)})^2 c_{1,k}^{l',k'} + d_{1,k}, \ldots, \sum_{l'=1}^L \sum_{k'=1}^K (\rho_{l',k'}^{(n-1)})^2 c_{L,k}^{l',k'} + d_{L,k} \right).
				\end{equation} %\vspace*{-0.5cm}
			\end{figure*}
			\item $\rho_{l,k}$ is updated as in \eqref{eq:rho_lkSol}.
			\begin{figure*}
				\begin{equation} \label{eq:rho_lkSol}
				\rho_{l,k}^{(n)} = \min \left(\frac{w_{l,k}^{(n)} \mathfrak{Re}  \left( u_{l,k}^{(n)}   \sum_{l'=1}^L (a_{l,k}^{l',(n)})^{\ast} b_{l,k}^{l'} \right) }{\sum_{l' =1}^L w_{l',k}^{(n)} |u_{l',k}^{(n)}|^2 \left| \sum_{l''=1}^L (a_{l',k}^{l'',(n)})^{\ast} b_{l,k}^{l''} \right|^2 +  \sum_{l'=1}^L \sum_{k'=1}^K w_{l',k'}^{(n)} |u_{l',k'}^{(n)}|^2 \sum_{l''=1}^L | a_{l',k'}^{l'',(n)}|^2 c_{l'',k'}^{l,k}}, \sqrt{P_{\max,l,k}} \right).
				\end{equation} %\vspace*{-0.5cm}
				\hrulefill
			\end{figure*}
		\end{itemize}
		If we denote the stationary point to \eqref{Problem: Sumratev3} that is attained by the above iterative algorithm as $n\to \infty$ by $u_{l,k}^{\mathrm{opt}}$, $w_{l,k}^{\mathrm{opt}}$, $\mathbf{a}_{l,k}^{\mathrm{opt}}$, and $\rho_{l,k}^{\mathrm{opt}}$, for all $l,k$, 
		then the solution $\{\mathbf{a}_{l,k}^{\mathrm{opt}}\}, \{(\rho_{l,k}^{\mathrm{opt}})^2\},$ is also a stationary point to the problem \eqref{Problem: Sumratev1}.
	\end{theorem}
	\begin{proof}
		The proof is available in Appendix~\ref{Appendix: Proof_Of_Theorem: IterativeSol}.
	\end{proof}
	The iterative algorithm is summarized in Algorithm~\ref{Algorithm:CentralizedApproach}. With initial data power values in the feasible set, the related \abb{LSFD} vectors are computed by using Corollary~\ref{corollary:Opt_LSFD}.\footnote{We observe faster convergence with a hierarchical initialization of $\rho_{l,k}^{(0)}, \forall l,k,$ than with an all-equal initialization. In simulations, we initialize $\rho_{l,k}^{(0)}, \forall l,k,$ as uniformly distributed over the range $\left[0, \sqrt{P_{\max,l,k}} \right]$.} After that, the iterative algorithm in Theorem~\ref{Theorem: IterativeSol} is used to obtain a stationary point to the sum \abb{SE} optimization problem~\eqref{Problem: Sumrate}.  Algorithm~\ref{Algorithm:CentralizedApproach} can be terminated when the variation between two consecutive iterations are small. In particular, for a given $\epsilon \geq 0$, the stopping criterion can, for instance, be defined as 
	\begin{equation} \label{eq:StopCriterion}
	\left| \sum_{l=1}^L \sum_{k=1}^K R_{l,k}^{(n)} - \sum_{l=1}^L \sum_{k=1}^K R_{l,k}^{(n-1)} \right| \leq \epsilon.
	\end{equation}
	Because all the update states in Algorithm~\ref{Algorithm:CentralizedApproach} are in closed form, for an initial point in the feasible set we can compute the exact number of arithmetic operations needed to obtain a given $\epsilon$-accuracy. For simplicity, let us only count complex multiplications, divisions, and logarithms, which are the main operations. Then, the number of arithmetic operations that Algorithm~\ref{Algorithm:CentralizedApproach} requires is
	\begin{multline} \label{eq:NumOfArithmeticPointv1}
	N \Big( 11 L^3 K^2 + 6 L^3 K + \frac{L^4K + 53 L^2 K }{3}  \\ 
	+ 3L^2K^2 + 16LK +2 \Big),
	\end{multline}
	where $N$ is the number of iterations to reach the stationary point. We obtain \eqref{eq:NumOfArithmeticPointv1} by assuming that a Cholesky decomposition is used to compute matrix inversion in \eqref{eq: a_lk_n1}.
	\begin{algorithm}[t]
		\caption{Alternating optimization approach for \eqref{Problem: Sumratev3}} \label{Algorithm:CentralizedApproach}
		\textbf{Input}: Maximum data powers $P_{\max,l,k}, \forall l,k$; Large-scale fading coefficients $\beta_{l,k}^{l'}, \forall, l,k,l'$. Initial coefficients $\rho_{l,k}^{(0)}, \forall l,k$.  Set up $n=1$ and  compute $\mathbf{a}_{l,k}^{(0)}, \forall l,k,$ using Corollary~\ref{corollary:Opt_LSFD}.
		\begin{itemize}
			\item[1.] \emph{Iteration} $n$:
			\begin{itemize}
				\item[1.1.] Update $u_{l,k}^{(n)}$ using \eqref{eq:u_lk_n1} where $\tilde{u}_{l,k}^{(n-1)}$ is computed as in \eqref{eq:tildeu_lk_n1}.
				\item[1.2.] Update $w_{l,k}^{(n)}$ using \eqref{eq:w_lk_n1} where $e_{l,k}^{(n)}$ is computed as in \eqref{eq:e_lk_n1}. 
				\item[1.3.] Update $\mathbf{a}_{l,k}^{(n)}$ by using \eqref{eq: a_lk_n1} where $\mathbf{C}_{l,k}^{(n-1)}$ is computed as in \eqref{eq:C_lk_n} and $\mathbf{b}_{l,k}$ as in \eqref{eq:bik}.
				\item[1.4.] Update $\rho_{l,k}^{(n)}$ by using \eqref{eq:rho_lkSol}.
			\end{itemize}
			\item[2.] If \textit{Stopping criterion \eqref{eq:StopCriterion} is satisfied $\rightarrow$ Stop}. Otherwise, go to Step 3.
			\item[3.] Store the currently best solution: $\rho_{l,k}^{(n)}$ and $\mathbf{a}_{l,k}^{(n)}$, $\forall l,k$. Set $n = n+1$, then go to Step $1$.
		\end{itemize}
		\textbf{Output}: The optimal solutions: $\rho_{l,k}^{\mathrm{opt}} = \rho_{l,k}^{(n)}, \mathbf{a}_{l,k}^{\mathrm{opt}}= \mathbf{a}_{l,k}^{(n)} \forall l,k.$
	\end{algorithm}
	\subsection{Sum \abb{SE} Optimization Without Using \abb{LFSD}}
	For completeness, 
	we also study a multi-cell Massive \abb{MIMO} system that only uses one-layer decoding. 
	This scenario is considered as a benchmark to investigate the improvements of our proposed joint data power control and \abb{LSFD} design in the previous section. 
	Mathematically, this is a special case of the above analysis, in which the elements of the \abb{LSFD} vector $\mathbf{a}_{l,k}, \forall l,k$ are defined as
	\begin{equation}
	a_{l,k}^{l'}  = \begin{cases}
	1, \quad \mathrm{for} \quad l' = l,\\ 
	0,  \quad \mathrm{for} \quad l' \neq l.
	\end{cases}
	\end{equation}
	With the \abb{LSFD} vector fixed, the \abb{SE} for each user in the network is a function only of the data power coefficients and it saturates when the number of \abb{BS} antennas is increased without bound.
	The \abb{SE} can, thus, only be improve through data power control. 
	For this communication scenario, the problem \eqref{Problem: Sumratev3} becomes
	\begin{equation} \label{Problem: Sumratev5}
	\begin{aligned}
	& \underset{ \substack{\{ \rho_{l,k} \geq 0 \}, \{ w_{l,k} \geq 0 \}, \{ u_{l,k} \} }}{\mathrm{minimize}}
	&&  \sum_{l=1}^L \sum_{k=1}^K w_{l,k} e_{l,k} - \ln (w_{l,k}) \\
	& \mathrm{subject\,\,to}
	& & \rho_{l,k}^2 \leq P_{\mathrm{max},l,k} \;, \forall l,k,\\
	\end{aligned}
	\end{equation}
	where $e_{l,k}$ is defined as
	\begin{equation} \label{eq:errorOptv2}
	\begin{split}
	e_{l,k} \triangleq& |u_{l,k}|^2 \left( \sum\limits_{l'=1 }^L \rho_{l',k}^2 |b_{l',k}^{l}|^2  + \sum\limits_{l'=1}^L \sum\limits_{k'=1}^K \rho_{l',k'}^2 c_{l,k}^{l',k'} + d_{l,k}  \right) \\
	&- 2 \rho_{l,k} \mathfrak{Re} \left(u_{l,k} b_{l,k}^{l} \right) + 1.
	\end{split}
	\end{equation}
	
	The alternating optimization approach in Algorithm~\ref{Algorithm:CentralizedApproach} can also be applied to the problem in \eqref{Problem: Sumratev5} to obtain a stationary point as shown in the following corollary.
	\begin{corollary} \label{Corollary: IterativeSol}
		A stationary point to \eqref{Problem: Sumratev5} is obtained by iteratively updating $\{ u_{l,k}, w_{l,k}, \rho_{l,k} \}$. At iteration $n$, these optimization parameters are updated as
		\begin{itemize}
			\item $u_{l,k}$ is updated as 
			\begin{equation} \label{eq:u_lk_n2}
			u_{l,k}^{(n)} = \frac{\rho_{l,k}^{(n-1)}  ( b_{l,k}^{l} )^{\ast} }{\tilde{u}_{l,k}^{(n-1)} },
			\end{equation}
			where $\tilde{u}_{l,k}^{(n-1)}$ is computed as
			\begin{multline} \label{eq:ulktilde}
			\tilde{u}_{l,k}^{(n-1)} = \sum\limits_{l'=1}^L (\rho_{l',k}^{(n-1)})^2 |b_{l',k}^{l} |^2 \\
			+ \sum\limits_{l'=1}^L \sum\limits_{k'=1}^K (\rho_{l',k'}^{(n-1)})^2 c_{l,k}^{l',k'} +  d_{l,k}.
			\end{multline}
			\item $w_{l,k}$ is updated as:
			\begin{equation} \label{eq:w_lk_n2}
			w_{l,k}^{(n)} = \left(e_{l,k}^{(n)}\right)^{-1},
			\end{equation}
			where $e_{l,k}^{(n)}$ is computed as
			\begin{equation} \label{eq:e_lk_n2}
			e_{l,k}^{(n)} =  |u_{l,k}^{(n)}|^2 \tilde{u}_{l,k}^{(n-1)}
			- 2 \rho_{l,k}^{(n-1)} \mathfrak{Re}\left( u_{l,k}^{(n)}   b_{l,k}^{l} \right)  + 1.
			\end{equation}
			
			\item $\rho_{l,k}$ is updated as 
			\begin{equation} \label{eq:rho_lkSoln2}
			\rho_{l,k}^{(n)} = \min \left(\tilde{\rho}_{l,k}^{(n)}, \sqrt{P_{\max,l,k}} \right),
			\end{equation}
			where
			\begin{equation}
			\hspace{-1.7em}\tilde{\rho}_{l,k}^{(n)} \triangleq \frac{w_{l,k}^{(n)} \mathfrak{Re}  \left( u_{l,k}^{(n)} b_{l,k}^{l} \right) }{\sum\limits_{l' =1}^L w_{l',k}^{(n)} |u_{l',k}^{(n)}|^2 \left| b_{l,k}^{l'} \right|^2 +  \sum\limits_{l'=1}^L \sum\limits_{k'=1}^K w_{l',k'}^{(n)} |u_{l',k'}^{(n)}|^2 c_{l',k'}^{l,k}}.
			\end{equation}
		\end{itemize}
	\end{corollary}
	After initializing the data power coefficients to a point in the feasible set, Corollary~\ref{Corollary: IterativeSol} provides closed-form expressions to update each variable in the optimization~\eqref{Problem: Sumratev5} iteratively. 
	This benchmark only treats the data powers as optimization variables, so it is a simplification of Algorithm~\ref{Algorithm:CentralizedApproach}. 
	
		\begin{figure*}[t]
		\begin{minipage}{0.48\textwidth}
			\centering
			\vspace*{0.5cm}
			\includegraphics[trim=0cm 0cm 0cm 0cm, clip=true, width=9.25cm]{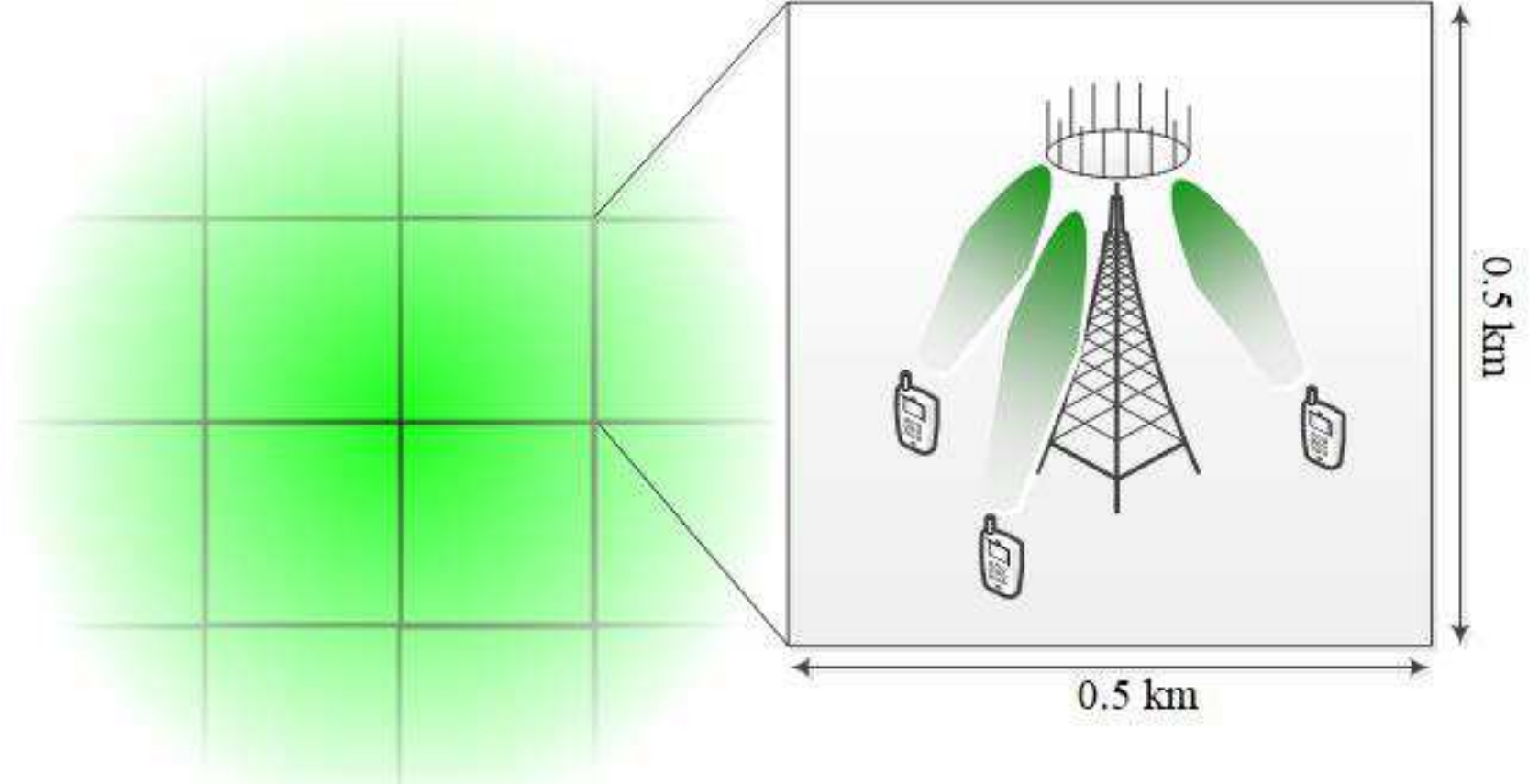} 
			\vspace*{-0.25cm}
			\caption{A wrapped-around cellular network used for simulation.}
			\label{Fig-SMMSEOptimized}
		\end{minipage}
		\hfill
		\begin{minipage}{0.48\textwidth}
			\centering
			\includegraphics[trim=0.5cm 0.1cm 1cm 0.5cm, clip=true, width=3in]{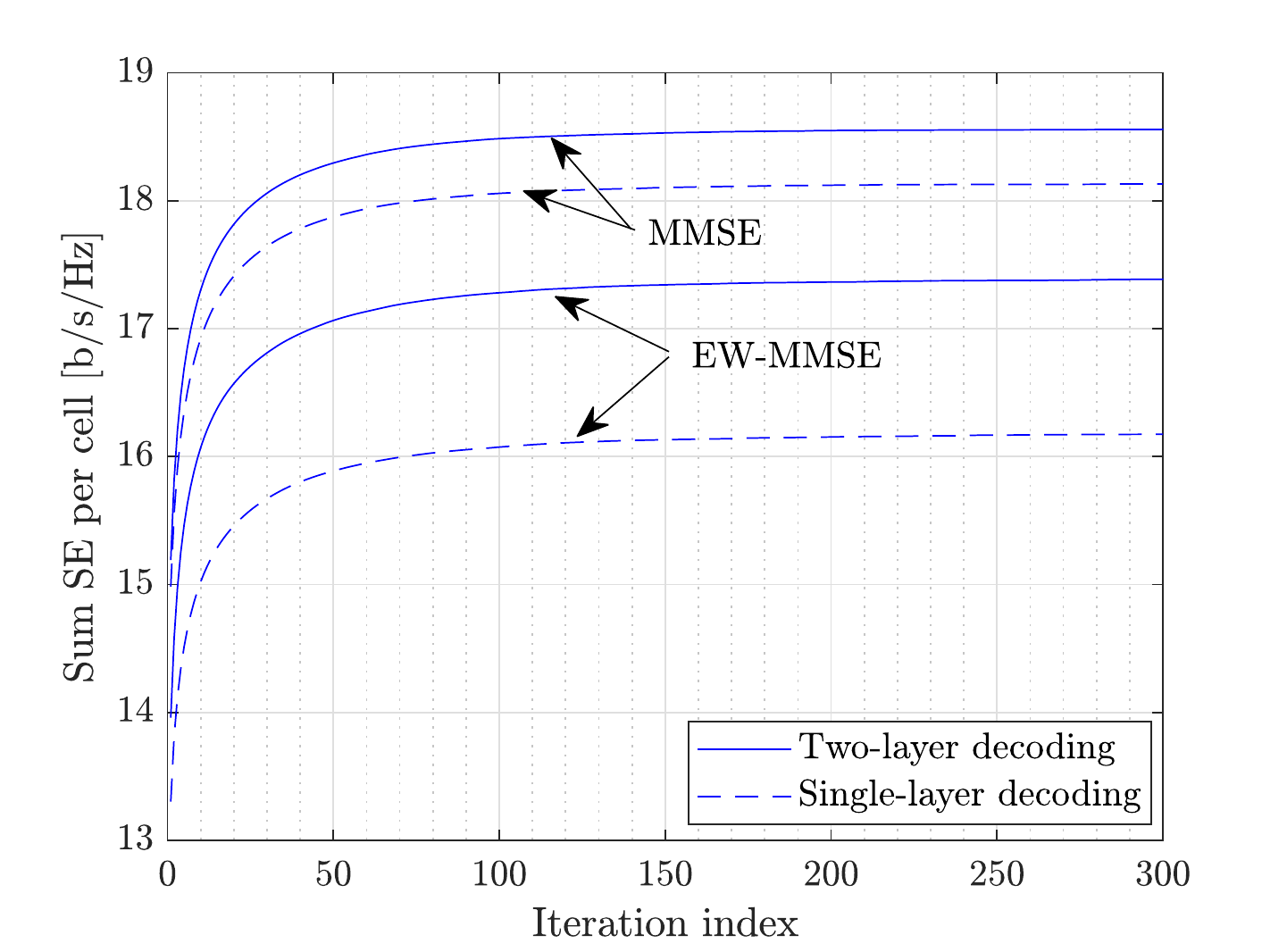} \vspace*{-0.15cm}
			\caption{Convergence of the proposed sum \abb{SE} optimization with $M=200$, $K= 5$, and $\varsigma = 0.8$.}
			\label{Fig-Convergence}
			\vspace*{-0.15cm}
		\end{minipage}
	\end{figure*}
	\begin{figure*}[t]
		\begin{minipage}{0.48\textwidth}
			\centering
			\includegraphics[trim=0.5cm 0.1cm 1cm 0.5cm, clip=true, width=3in]{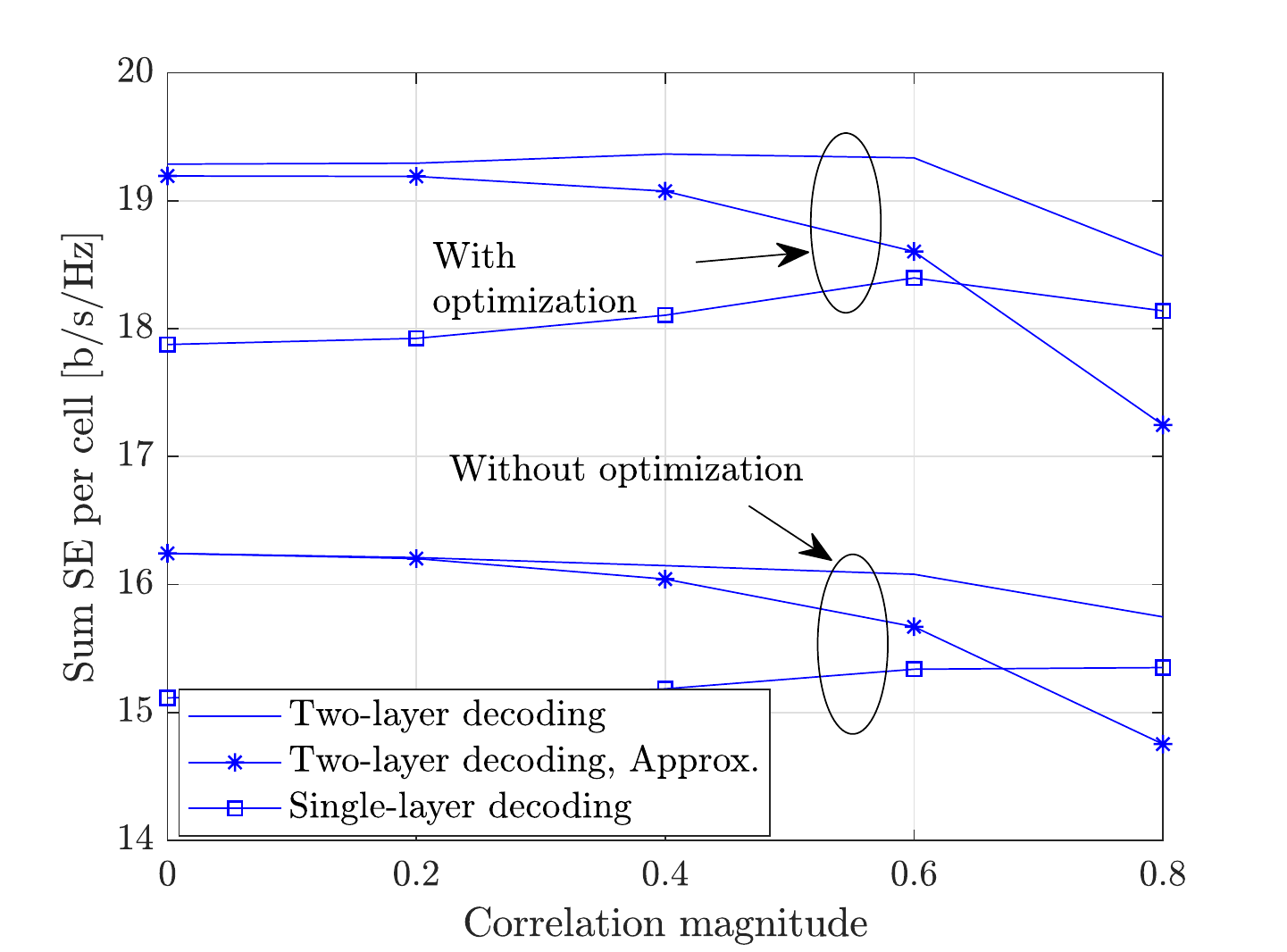} \vspace*{-0.15cm}
			\caption{Sum \abb{SE} per cell [b/s/Hz] versus different correlation magnitudes. The network uses \abb{MMSE} estimation, $M=200$, and $K= 5$.}
			\label{Fig-MMSE-DiffCofact}
			\vspace*{-0.15cm}
		\end{minipage}
		\hfill
		\begin{minipage}{0.48\textwidth}
			\centering
			\includegraphics[trim=0.5cm 0.1cm 1cm 0.5cm, clip=true, width=3in]{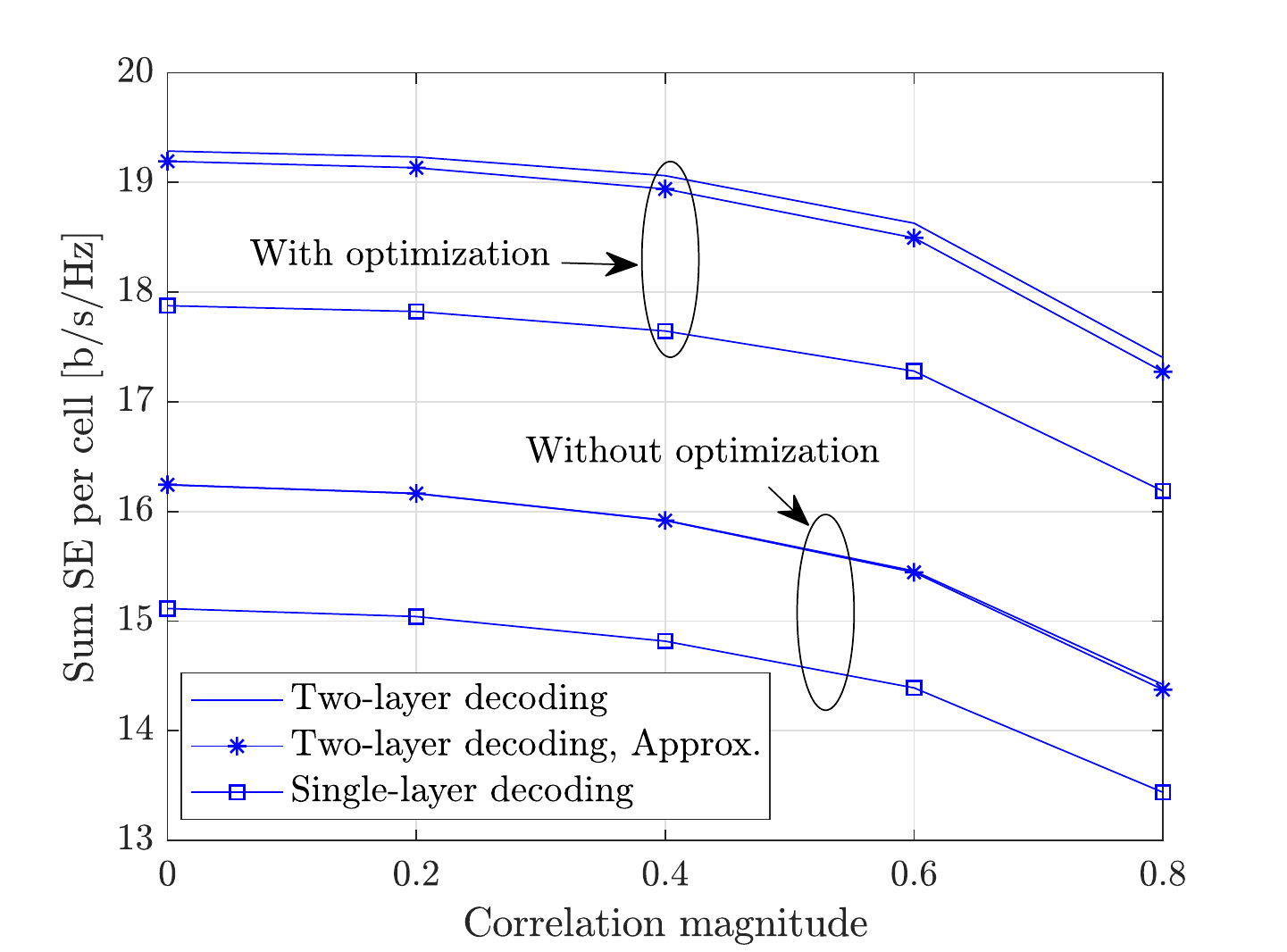} \vspace*{-0.15cm}
			\caption{Sum \abb{SE} per cell [b/s/Hz] versus different correlation magnitudes. The network uses \abb{EW-MMSE} estimation, $M=200$, and $K= 5$.}
			\label{Fig-SMMSE-DiffCofact}
			\vspace*{-0.15cm}
		\end{minipage}
	\end{figure*}
	\section{Numerical Results} \label{Section:NumericalResults}
	%\vspace*{-0.25cm}
	To demonstrate the effectiveness of the proposed algorithms, we consider a wrapped-around cellular network with four cells as illustrated in Fig.~\ref{Fig-SMMSEOptimized}. The distance between user~$k$ in cell~$l'$ and BS~$l$ is denoted by $d_{l',k}^l$.  
	The users in each cell are uniformly distributed over the cell area that is at least 35\,m away from the BS, i.e.\ $d_{l',k}^l \geq \text{35\,m}$.  
	Monte-Carlo simulations are done over $300$ random sets of user locations, for almost figures, but Fig.~\ref{Fig-OtherLinears} is obtained by 3000 random sets of user locations  which the moments of complex Gaussian distributions  are computed by 1000 random realizations of small-scale fading, 
	
	We model the system parameters and large-scale fading similar to the 3GPP LTE specifications \cite{LTE2010b}. The system uses $20$\,MHz of bandwidth, the noise variance is $-96$\,dBm, and the noise figure is $5$\,dB.  The large-scale fading coefficient $\beta_{l,k}^{l'}$ in decibel is computed as
	\begin{equation}
	\left[ \beta_{l,k}^{l'} \right]_{\text{dB}} = -148.1 - 37.6 \log_{10} \left( d_{l,k}^{l'} / 1\,\text{km}  \right)+ z_{l,k}^{l'},
	\end{equation}
	where the decibel value of the shadow fading, $z_{l,k}^{l'}$, has a Gaussian distribution with zero mean and standard derivation~$7$.  The spatial correlation matrix of the channel from user~$k$ in cell~$l$ to BS~$l'$ is described by the exponential correlation model, which models a uniform linear array \cite{Loyka2001a}:
	\begin{equation}
	\mathbf{R}_{l,k}^{l'} = \beta_{l,k}^{l'} \begin{bmatrix}
	1     &  r_{l,k}^{l',\ast}   &  \cdots     &  (r_{l,k}^{l',\ast})^{M-1} \\
	r_{l,k}^{l'}        & 1     & \cdots    & (r_{l,k}^{l',\ast})^{M-2} \\
	\vdots & \vdots &  \ddots      &  \vdots  \\
	(r_{l,k}^{l'})^{M-1}     & (r_{l,k}^{l'})^{M-2}    &  \cdots   &  1
	\end{bmatrix},
	\end{equation}
	where the correlation coefficient $r_{l,k}^{l'} = \varsigma e^{j\theta_{l,k}^{l'}}$, the correlation magnitude $\varsigma$ is in the range $[0,1]$ and the user incidence angle to the array boresight is $\theta_{l,k}^{l'}$.  
	
		\begin{figure*}[t]
		\begin{minipage}{0.48\textwidth}
			\centering
			\includegraphics[trim=0.5cm 0.1cm 1cm 0.5cm, clip=true, width=3in]{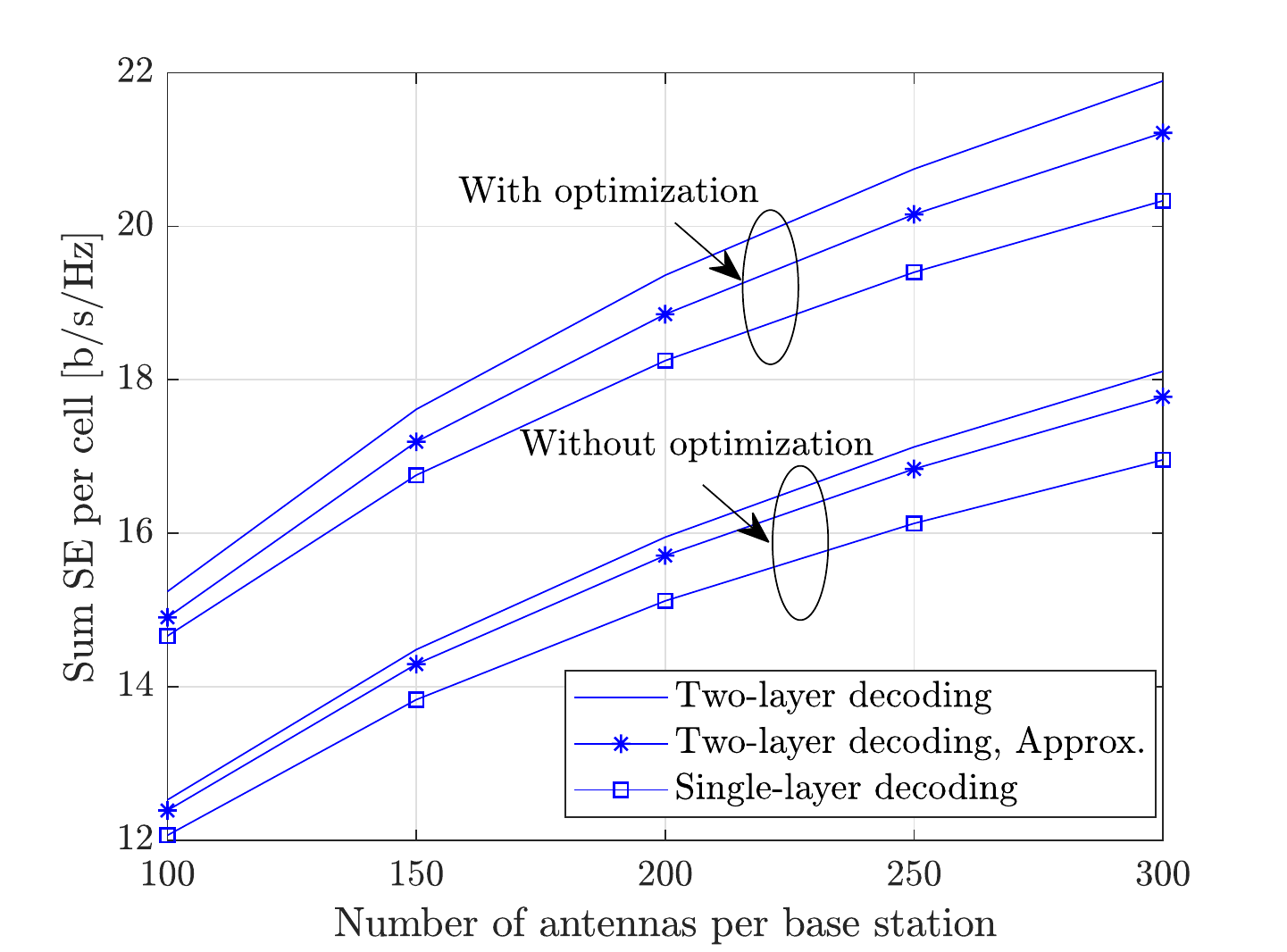} \vspace*{-0.15cm}
			\caption{Sum \abb{SE} per cell [b/s/Hz] versus different number of BS antennas. The network uses \abb{MMSE} estimation, $K= 5$, and $\varsigma = 0.5$.}
			\label{Fig-AntennaMMSE}
			\vspace*{-0.15cm}
		\end{minipage}
		\hfill
		\begin{minipage}{0.48\textwidth}
			\centering
			\includegraphics[trim=0.5cm 0.1cm 1cm 0.5cm, clip=true, width=3in]{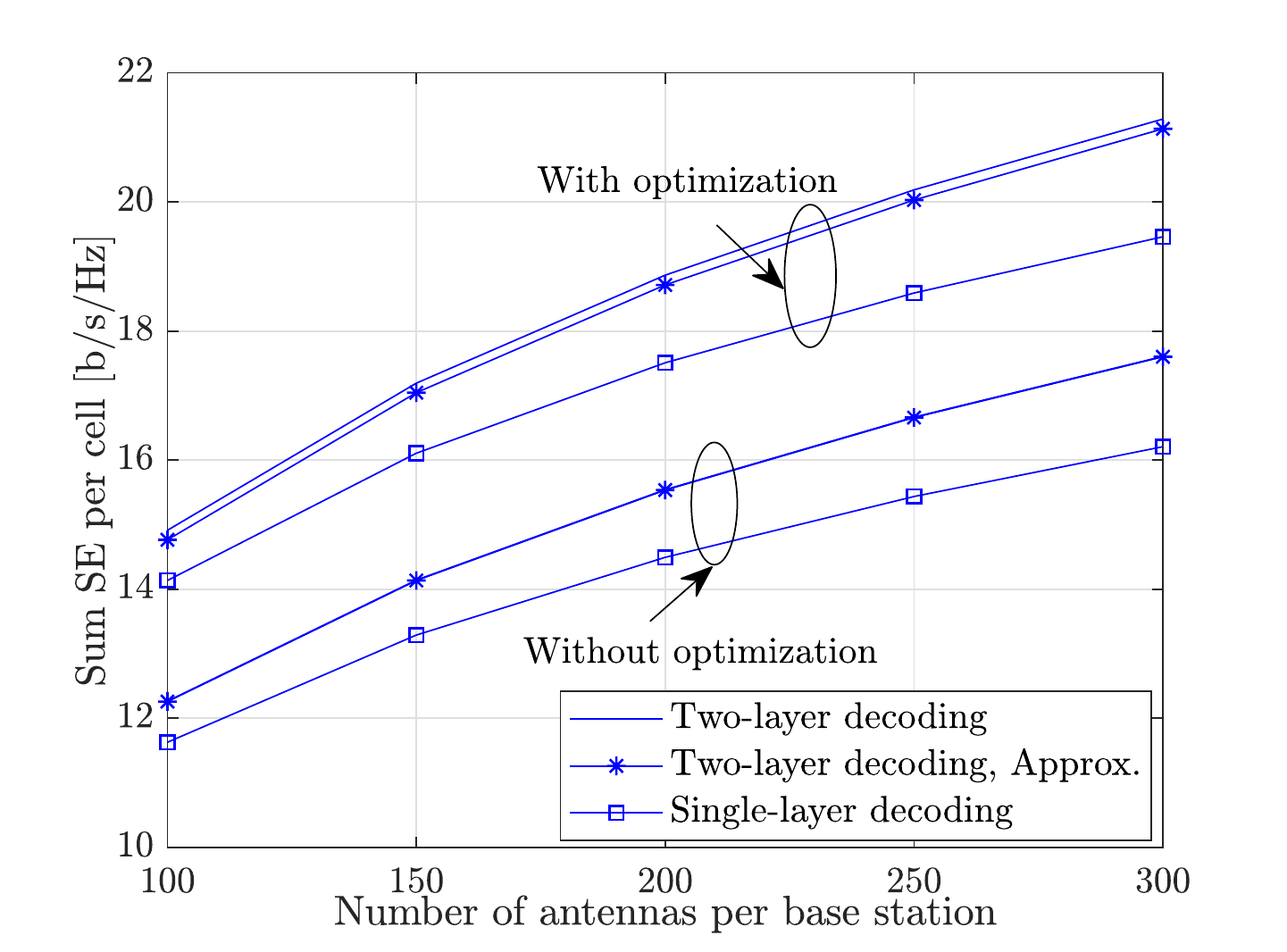} \vspace*{-0.15cm}
			\caption{Sum \abb{SE} per cell [b/s/Hz] versus different number of BS antennas. The network uses \abb{EW-MMSE} estimation, $K= 5$, and $\varsigma = 0.5$.}
			\label{Fig-AntennaSMMSE}
			\vspace*{-0.15cm}
		\end{minipage}
	\end{figure*}
	\begin{figure*}[t]
		\begin{minipage}{0.48\textwidth}
			\centering
			\includegraphics[trim=0.5cm 0.1cm 1cm 0.5cm, clip=true, width=3in]{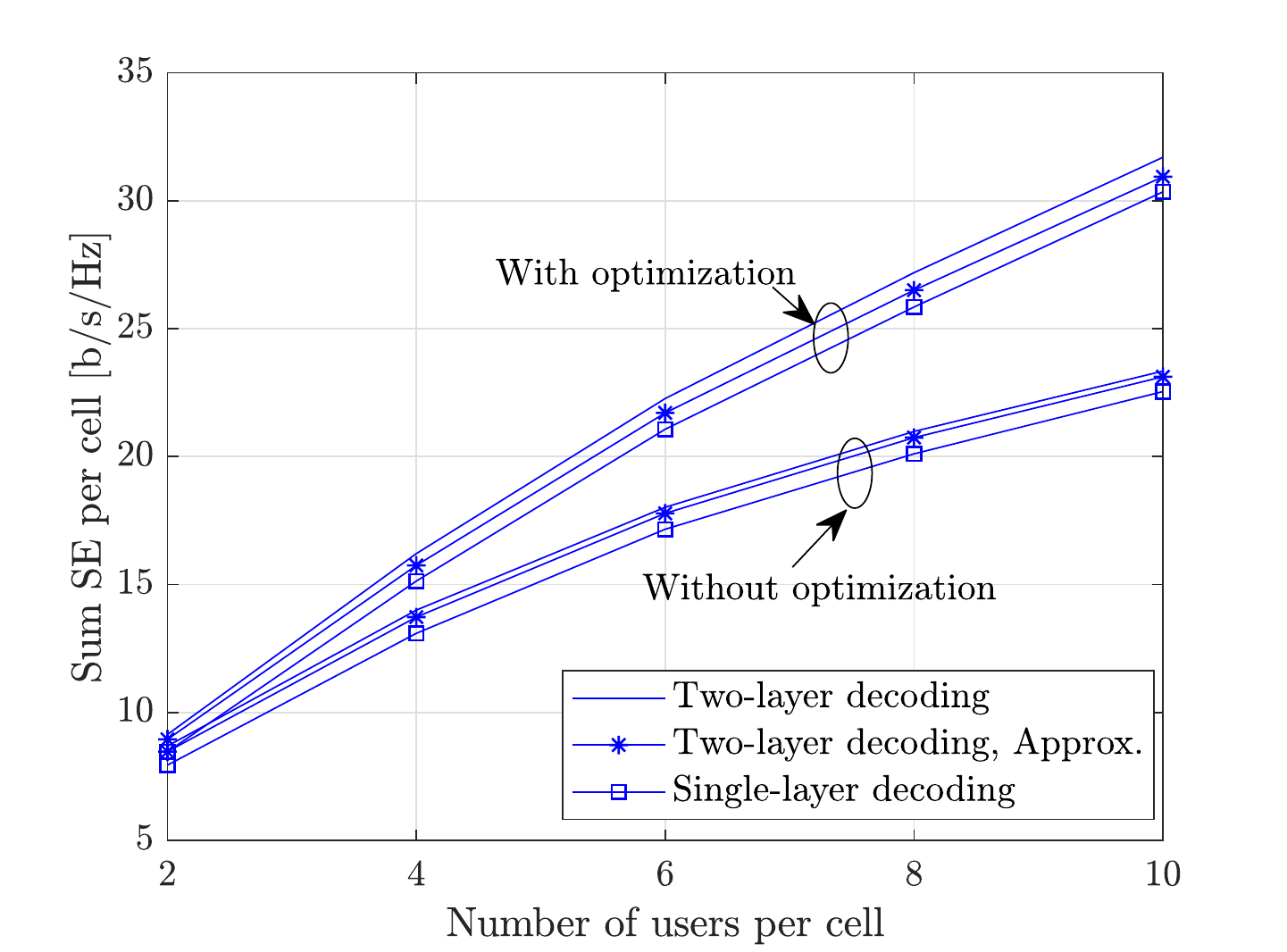} \vspace*{-0.15cm}
			\caption{Sum \abb{SE} per cell [b/s/Hz] versus different number of users per cell. The network uses \abb{MMSE} estimation, $M= 200$, and $\varsigma = 0.5$.}
			\label{Fig-UserMMSE}
			\vspace*{-0.15cm}
		\end{minipage}
		\hfill
		\begin{minipage}{0.48\textwidth}
			\centering
			\includegraphics[trim=0.5cm 0.1cm 1cm 0.5cm, clip=true, width=3in]{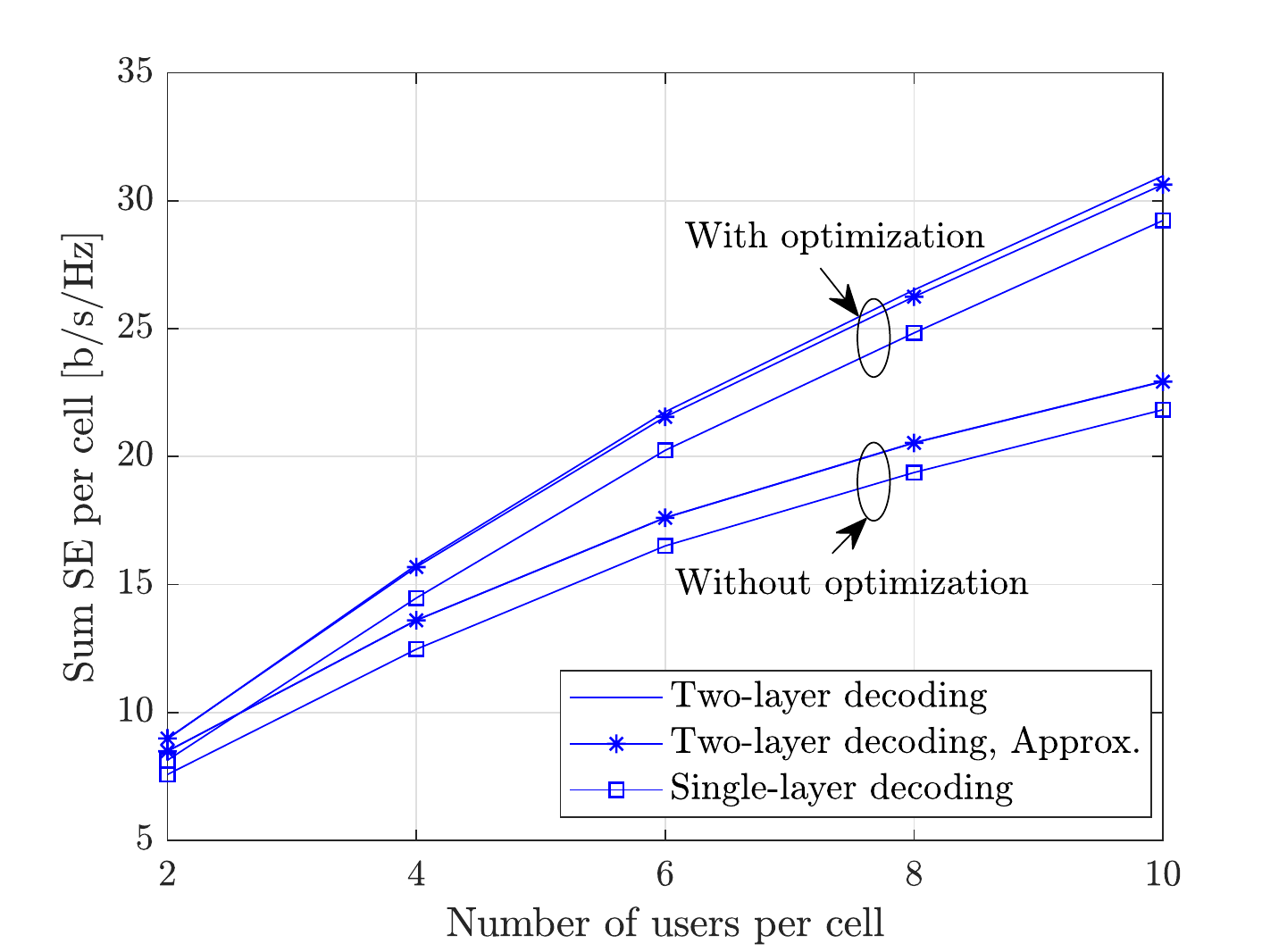} \vspace*{-0.15cm}
			\caption{Sum \abb{SE} per cell [b/s/Hz] versus different number of users per cell. The network uses \abb{EW-MMSE} estimation, $M= 200$, and $\varsigma = 0.5$.}
			\label{Fig-UserSMMSE}
			\vspace*{-0.15cm}
		\end{minipage}
	\end{figure*}

	We assume that the power is fixed to $200$\,mW for each pilot symbol and it is also the maximum power that each user can allocate to a data symbol, i.e., $P_{\max,l,k} = 200$\,mW. 
	Extensive numerical results will be presented from the following methods with either \abb{MMSE} or \abb{EW-MMSE} estimation:
	\begin{itemize}
		\item[\textit{(i)}] \textit{Single-layer decoding system with fixed data power}: Each \abb{BS} uses \abb{MRC} for data decoding for the users in the own cell, and all users transmit data symbols with the same power $200$\,mW.  
		
		\item[\textit{(ii)}] \textit{Single-layer decoding system with data power control}: This benchmark is similar to \textit{(i)}, but the data powers are optimized using the weighted \abb{MMSE} algorithm in Corollary~\ref{Corollary: IterativeSol}.
		
		\item[\textit{(iii)}] \textit{Two-layer decoding system with fixed data power and \abb{LSFD} vectors}: The network deploys the two-layer decoding as shown in Fig.~\ref{fig:decoder}, using \abb{MRC} and \abb{LSFD}. The data symbols have fixed power $200$\,mW and the \abb{LSFD} vectors are computed using Corollary~\ref{corollary:Opt_LSFD}. 
		
		\item[\textit{(iv)}]  \textit{Two-layer decoding system with fixed data power and approximate \abb{LSFD} vectors}: This benchmark is similar to \textit{(iii)}, but the \abb{LSFD} vectors are computed using only the diagonal elements of the channel correlation matrices. This allows us to study how inaccurate \abb{LSFD} vectors degrade sum \abb{SE}. 
		
		\item[\textit{(v)}] \textit{Two-layer decoding system with optimized data power and  \abb{LSFD} vectors}: This benchmark is similar to \textit{(iii)}, but the data powers and \abb{LSFD} vectors are computed using the weighted \abb{MMSE} algorithm as in Theorem~\ref{Theorem: IterativeSol}.
		
		\item[\textit{(vi)}] \textit{Two-layer decoding system with optimized data power and approximate \abb{LSFD} vectors}: This benchmark is similar to \textit{(v)}, but the \abb{LSFD} vectors are computed by Corollary~\ref{corollary:Opt_LSFD} based on only the diagonal coefficients of the channel correlation matrices.
	\end{itemize}

	\subsection{Convergence}
	%\vspace*{-0.2cm}
	Fig.~\ref{Fig-Convergence} shows the convergence of the proposed methods for sum \abb{SE} optimization in Theorem~\ref{Theorem: IterativeSol} and Corollary~\ref{Corollary: IterativeSol} for both \abb{MMSE} and \abb{EW-MMSE} estimation. 
	From the initial data powers, in the feasible set, updating the optimization variables gives improved sum \abb{SE} in every iteration. 
	For a system that uses \abb{MMSE} estimation and \abb{LSFD}, the sum \abb{SE} per cell is about $22.2\%$ better at the stationary point than at the initial point. 
	The corresponding improvement for the system that uses \abb{EW-MMSE} estimation is about $24.7\%$. 
	By using \abb{MMSE} estimation, the two-layer decoding system gives $2.4\%$ better sum \abb{SE} than a system with single-layer decoding. 
	The corresponding gain for \abb{EW-MMSE} estimation is up to $7.5\%$. 
	Besides, \abb{MMSE} estimation gives an \abb{SE} that is up to $12.1\%$ higher than \abb{EW-MMSE}.  
	The proposed optimization methods need around $100$ iterations to converge, but the complexity is low since every update in the algorithm consists of evaluating a closed-form expression.
	
	The approximation in $(vi)$ of the channel correlation matrix as diagonal breaks the convergence statement in Theorem~\ref{Theorem: IterativeSol}, so it is not included in Fig.~\ref{Fig-Convergence}. 
	Hereafter, when we consider $(vi)$ for comparison, we select the highest sum \abb{SE} among $500$ iterations.
	
	\subsection{Impact of Spatial Correlation}
	%\vspace*{-0.2cm}
	Figs.~\ref{Fig-MMSE-DiffCofact} and \ref{Fig-SMMSE-DiffCofact} show the sum \abb{SE} per cell as a function of the channel correlation magnitude $\varsigma$ for a multi-cell Massive \abb{MIMO} system using either \abb{MMSE} or \abb{EW-MMSE} estimation. 
	First, we observe the large gains in sum \abb{SE} attained by using \abb{LSFD} detection. 
	With \abb{MMSE} estimation, the sum \abb{SE} increases with up to $7.5\%$ in the case of equally fixed data powers, while that gain is about $7.9\%$ for jointly optimizing data powers and \abb{LSFD} vectors. 
	The same maximum gains are observed when using \abb{EW-MMSE} estimation, since these gains occur when $\varsigma = 0$ (in which case \abb{MMSE} and \abb{EW-MMSE} coincide). 
	The performance of \abb{EW-MMSE} estimation is worse than that of \abb{MMSE} when the correlation magnitude is increased, because \abb{EW-MMSE} does not use the knowledge of the spatial correlation to improve the estimation quality. For example, \abb{MMSE} estimation obtains $6.68\%$ and $9.91\%$ higher SE than \abb{EW-MMSE} with and without data power control, respectively. The advantage of \abb{EW-MMSE} is the reduced computational complexity.
	
	Interestingly, Figs.~\ref{Fig-MMSE-DiffCofact} and \ref{Fig-SMMSE-DiffCofact} indicate that it is sufficient to use only the large-scale fading coefficients when constructing the \abb{LSFD} vectors in many scenarios. 
	Specifically, in the system with \abb{EW-MMSE} estimation, the approximate \abb{LSFD} vectors yield almost the same sum \abb{SE} as the optimal ones.
	Meanwhile, in the case of \abb{MMSE} estimation, the loss from the approximation of LSFD vectors, which are only based on the diagonal values of channel correlation matrices, grows up to $6.7\%$ when having a correlation magnitude of $0.8$. In comparison to \abb{MMSE}, increasing the spatial correlation does not improve the performance of the approximate \abb{LFSD} vectors when using \abb{EW-MMSE} estimation, since it does not utilize the spatial correlation in the estimation phase. Consequently, \abb{MMSE} and \abb{EW-MMSE}  perform almost equally with the maximum difference $0.71\%$.
	
	Moreover, the performance is greatly improved when the data powers are optimized. 
	The gain varies from $17.9\%$ to $20.7\%$. 
	The gap becomes larger as the channel correlation magnitude increases. 
	This shows the importance of doing joint data power control and \abb{LSFD} optimization in Massive \abb{MIMO} systems with spatially correlated channels.

	\subsection{Impact of Number of Antennas and Users}
	%\vspace*{-0.2cm}
	Figs.~\ref{Fig-AntennaMMSE} and \ref{Fig-AntennaSMMSE} show the sum \abb{SE} per cell as a function of the number of BS antennas with \abb{MMSE} and \abb{EW-MMSE} estimation, respectively. 
	Two-layer decoding gives improvements in all the cases. 
	In case of \abb{MMSE} estimation, by increasing the number of BS antennas from $100$ to $300$, the gain of using \abb{LSFD} increases from $4.0\%$ to $7.7\%$ with optimized data power, and from $3.8\%$ to $6.8\%$ with equal data power. 
	In case of \abb{EW-MMSE} estimation and fixed transmitted power level, \abb{LSFD} increase the sum \abb{SE} by $5.5\%$ to $8.6\%$ compared to using only \abb{MRC}. 
	Besides, by optimizing the data powers, the gain from using \abb{LSFD} is between $5.5\%$ and $9.4\%$. Among all considered scenarios, \abb{MMSE} estimation provides up to $4.6\%$ higher sum \abb{SE} than \abb{EW-MMSE}.

	Figs.~\ref{Fig-UserMMSE} and \ref{Fig-UserSMMSE} show the sum \abb{SE} per cell as a function of the number of users per cell when using \abb{MMSE} and \abb{EW-MMSE} estimation, respectively. 
	The figures demonstrate how the gain from power control increases with the number of users. 
	The gain grows from  $5.2\%$ for two users to $35.8\%$ for for ten users.
	The approximated version of \abb{LSFD} detection works properly in all tested scenarios, in the sense that the maximum loss in \abb{SE} is only up to $2.9\%$. In these figures, \abb{MMSE} provides up to $5\%$ higher SE than \abb{EW-MMSE}.
	
	\begin{figure}[t]
		\centering
		\includegraphics[trim=0.5cm 0.0cm 1cm 0.5cm, clip=true, width=3in]{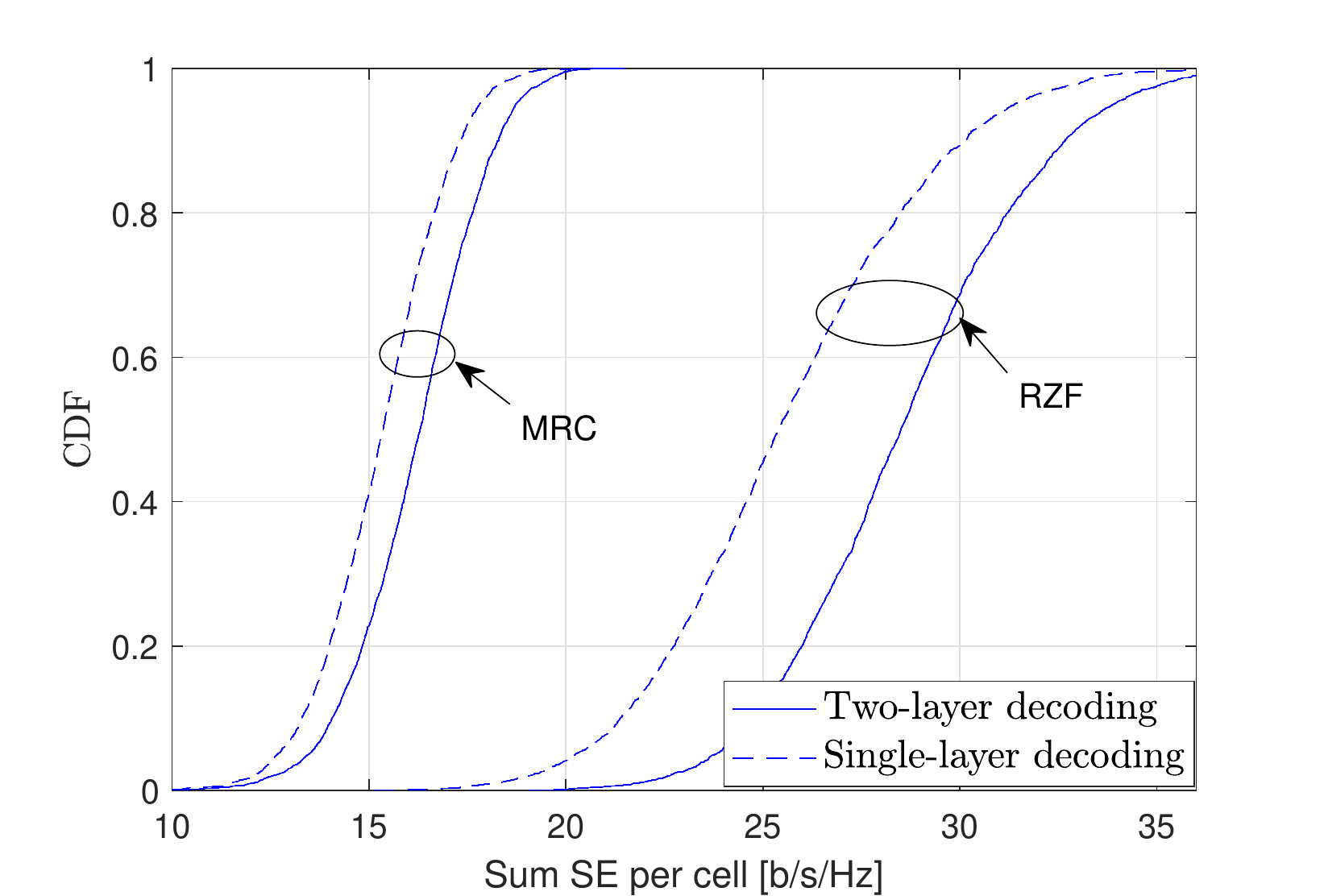} \vspace*{-0.15cm}
		\caption{\abb{CDF} of sum \abb{SE} per cell [b/s/Hz] for \abb{MRC} and \abb{RZF}. The network uses $M=200, K =5,\varsigma = 0.5,$ and \abb{MMSE} estimation.}
		\label{Fig-OtherLinears}
		\vspace*{-0.15cm}
	\end{figure}
	%\vspace*{-0.5cm}
	\subsection{Performance of Regularized Zero-Forcing}
	%	\vspace*{-0.2cm}
	Fig.~\ref{Fig-OtherLinears} compares the cumulative distribution function (\abb{CDF}) of the sum \abb{SE} per cell when using either \abb{MRC} or \abb{RZF} in the first layer. \abb{MMSE} estimation is used for channel estimation. An equal pilot and data power of $200$ mW is allocated to each transmitted symbol. We first observe that  \abb{RZF} achieves much higher SE than  \abb{MRC}. The performance gain is $67.2\%$ and $76.6\%$ with single-layer and two-layer decoding, respectively. Because \abb{RZF} cancels non-coherent interference effectively at the first layer, the second layer can improve the sum \abb{SE} by $11.80\%$. Meanwhile, the improvement is only $5.84\%$ if \abb{MRC} is used. 
	
	\section{Conclusion} \label{Section:Conclustion}
	%\vspace*{-0.2cm}
	This paper has investigated the performance of the \abb{LSFD} design in mitigating mutual interference for multi-cell Massive \abb{MIMO} systems with correlated Rayleigh fading. 
	This decoding design is deployed as a second decoding layer to mitigate the interference that remains after classical linear decoding. 
	Numerical results demonstrate the effectiveness of the \abb{LSFD} in reducing pilot contamination with the improvement of sum \abb{SE} for each cell up to about $10\%$ in the tested scenarios. 
	We have also investigated joint data power control and \abb{LSFD} design, which efficiently improves the sum \abb{SE} of the network. 
	Even though the sum \abb{SE} optimization is a non-convex and NP-hard problem, we proposed an iterative approach to obtain a stationary point with low computational complexity. 
	Numerical results showed improvements of sum \abb{SE} for each cell up to more than $20\%$ with using the limited number of BS antennas.
	\appendices
	\section{Useful Lemma and Definition}
	%\vspace*{-0.2cm}
	\begin{lemma}[Lemma~$2$ in \cite{Bjornson2015b}] \label{Lemma:Supplementary}
		Let a random vector be distributed as $\mathbf{u} \sim \mathcal{CN}(\mathbf{0}, \pmb{\Lambda})$ and consider an arbitrary, deterministic matrix $\mathbf{M}$. It holds that
		\begin{equation}
		\Exp \{ |\mathbf{u}^\conjtr \mathbf{M} \mathbf{u} |^2 \} = |\mathrm{tr}(\pmb{\Lambda} \mathbf{M} )|^2 + \mathrm{tr} (\pmb{\Lambda} \mathbf{M} \pmb{\Lambda} \mathbf{M}^\conjtr ). 
		\end{equation}
	\end{lemma}
	\begin{definition}[Stationary point, \cite{yang2017a}]  \label{Def:StationaryPoint}
		\footnote{Definition~\ref{Def:StationaryPoint} guarantees the existence of at least one stationary point for any non-convex problem as long as the feasible set is convex. The stationary point is even the global optimum if \eqref{Problem:GeneralOpti} is a convex problem.}
		Consider the optimization problem
		\begin{equation} \label{Problem:GeneralOpti}
		\begin{aligned}
		& \underset{ \mathbf{x}\in \mathcal{X} }{\mathrm{minimize}}
		&&  g(\mathbf{x}),
		\end{aligned}
		\end{equation}
		where the feasible set $\mathcal{X}$ is convex and $g(\mathbf{x}): \mathbb{R}^n \rightarrow \mathbb{R}$ is differentiable. 
		A point $\mathbf{y} \in \mathcal{X}$ is a stationary point to the optimization problem~\eqref{Problem:GeneralOpti} if the following property is true for all $\mathbf{x} \in \mathcal{X}$:
		\begin{equation}
		(\mathbf{x} - \mathbf{y})^\tr \nabla g(\mathbf{x})|_{\mathbf{y}} \geq 0.
		\end{equation}
		Note that a stationary point $\mathbf{y}$ of $g(\mathbf{x})$ can be obtained by solving the equation $\nabla g (\mathbf{x})  =\mathbf{0}$.
	\end{definition}
	
	\section{Proof of Theorem~\ref{Theorem1v1}} \label{Appendix: Proof-Theorem-v1}
	%\vspace*{-0.2cm}
	The numerator of \eqref{eq:GeneralSINR} is reformulated into
	\begin{equation}
	\Exp\{ |\mathtt{DS}_{l,k}|^2 \} = p_{l,k} | \mathbf{a}_{l,k}^\conjtr \mathbf{b}_{l,k}|^2.
	\end{equation}
	Meanwhile, the pilot contamination term in the denominator of \eqref{eq:GeneralSINR} is rewritten as
	\begin{equation}
	\begin{split}
	\Exp \{ |\mathtt{PC}_{l,k}|^2 \} &= \sum_{\substack{l''=1 \\ l'' \neq l }}^L p_{l'',k} \left|\sum_{l'=1}^L (a_{l,k}^{l'})^{\ast} \Exp \{\mathbf{v}_{l',k}^\conjtr \mathbf{h}_{l'',k}^{l'} \} \right|^2  \\ 
	& = \sum_{\substack{l''=1 \\ l'' \neq l }}^L p_{l'',k} \mathbf{a}_{l,k}^\conjtr \mathbf{b}_{l'',k} \mathbf{b}_{l'',k}^\conjtr \mathbf{a}_{l,k} \\
	& =  \mathbf{a}_{l,k}^\conjtr \mathbf{C}_{l,k}^{(1)} \mathbf{a}_{l,k}.
	\end{split}
	\end{equation}
	The beamforming gain uncertainty term in the denominator of \eqref{eq:GeneralSINR} is rewritten as
	\begin{equation}
	\begin{split}
	\Exp \{ | \mathtt{BU}_{l,k} |^2 \} = \mathbf{a}_{l,k}^\conjtr \mathbf{C}_{l,k}^{(2)} \mathbf{a}_{l,k}.
	\end{split}
	\end{equation}
	Similarly, the non-coherent interference term in the denominator is computed as
	\begin{equation}
	\begin{split}
	&\sum_{l'=1}^L \sum_{\substack{k'=1\\ k' \neq k}}^K \Exp \{ | \mathtt{NI}_{l',k'} |^2 \} \\
	&\quad= \sum_{l'=1}^L \sum_{\substack{k'=1\\ k' \neq k}}^K \left|  \sum_{l''=1}^L (a_{l,k}^{l''})^{\ast} \sqrt{p_{l',k'}} \mathbf{v}_{l'',k}^\conjtr \mathbf{h}_{l',k'}^{l''} \right|^2 \\
	&\quad= \mathbf{a}_{l,k}^H \mathbf{C}_{l,k}^{(3)} \mathbf{a}_{l,k}
	\end{split}
	\end{equation}
	and the additive noise term is computed as
	\begin{equation}
	\begin{split}
	\Exp \{ | \mathtt{AN}_{l,k} |^2 \} &=        \Exp \left\{ \left| \sum_{l'=1}^L (a_{l,k}^{l'})^{\ast} \mathbf{v}_{l',k}^\conjtr \mathbf{n}_{l'} \right|^2 \right\} \\
	&= \mathbf{a}_{l,k}^H \mathbf{C}_{l,k}^{(4)} \mathbf{a}_{l,k}.
	\end{split}
	\end{equation}
	The lower-bound on the uplink capacity given in Lemma~\ref{lemma:General_Rate} is written as
	\begin{equation} \label{eq:RayleighQuotientRate}
	R_{l,k} = \left( 1- \frac{\tau_\text{p}}{\tau_\text{c}} \right) \log_2 \left(1 + \frac{p_{l,k} |\mathbf{a}_{l,k}^\conjtr \mathbf{b}_{l,k} |^2}{\mathbf{a}_{l,k}^\conjtr  \left(\sum\limits_{i=1}^{4}\mathbf{C}_{l,k}^{(i)} \right) \mathbf{a}_{l,k}} \right).
	\end{equation}
	Since the SINR expression in \eqref{eq:RayleighQuotientRate} is a \emph{generalized Rayleigh quotient} with respect to $\mathbf{a}_{l,k}$, we can apply \cite[Lemma~B.10]{Bjornson2017bo} to obtain the maximizing vector $\mathbf{a}_{l,k}$ as in \eqref{eq:LSFDVector}. 
	Hence, using \eqref{eq:LSFDVector} in \eqref{eq:RayleighQuotientRate}, maximizes the \abb{SE} for both \abb{MMSE} and \abb{EW-MMSE} estimation.  The maximum \abb{SE} is given by \eqref{eq:OptimalRate} in the theorem.
	
	\section{Proof of Theorem~\ref{Theorem1} in case of MMSE} \label{Appendix: Proof-Theorem-1}
	%\vspace*{-0.2cm}
	Here the expectations in \eqref{eq:GeneralSINR} are computed.  The numerator of \eqref{eq:GeneralSINR}, $\Exp\{ |\mathtt{DS}_{l,k}|^2 \}$, becomes
	\begin{equation} \label{eq:SENumberatorMMSE}
	\begin{split}
	& p_{l,k} \left| \sum_{l'=1}^L (a_{l,k}^{l'})^{\ast} \Exp \left\{ \sqrt{\hat{p}_{l',k}} (\mathbf{R}_{l',k}^{l'} \pmb{\Psi}_{l',k}^{-1} \mathbf{y}_{l',k})^\conjtr \mathbf{h}_{l,k}^{l'} \right\} \right|^2 \\
	&= p_{l,k} \left| \sum_{l'=1}^L (a_{l,k}^{l'})^{\ast} \sqrt{\hat{p}_{l',k}}  \mathrm{tr} \left(\pmb{\Psi}_{l',k}^{-1} \mathbf{R}_{l',k}^{l'} \Exp \{ \mathbf{h}_{l,k}^{l'} \mathbf{y}_{l',k}^\conjtr \} \right) \right|^2 \\
	&=  p_{l,k} \left| \sum_{l'=1}^L (a_{l,k}^{l'})^{\ast}  \tau_{\text{p}} \sqrt{\hat{p}_{l',k} \hat{p}_{l,k} }  \mathrm{tr} \left(\pmb{\Psi}_{l',k}^{-1} \mathbf{R}_{l',k}^{l'} \mathbf{R}_{l,k}^{l'} \right) \right|^2 \\
	&= \tau_{\text{p}} p_{l,k} \left| \sum_{l'=1}^L (a_{l,k}^{l'})^{\ast}  b_{l,k}^{l'} \right|^2,
	\end{split}
	\end{equation}
	The variance of the pilot contamination in the denominator of \eqref{eq:GeneralSINR} is computed as 
	\begin{equation} \label{eq:MMSEPilotContamination}
	\begin{split}
	&  \sum_{\substack{l''=1 \\ l'' \neq l }}^L p_{l'',k} \left| \sum_{l'=1}^L (a_{l,k}^{l'})^{\ast}  \tau_\text{p} \sqrt{\hat{p}_{l',k} \hat{p}_{l'',k} } \mathrm{tr} \left( \pmb{\Psi}_{l',k}^{-1} \mathbf{R}_{l',k}^{l'} \mathbf{R}_{l'',k}^{l'} \right) \right|^2 \\
	&=  \tau_\text{p} \sum_{\substack{l''=1 \\ l'' \neq l }}^L p_{l'',k} \left| \sum_{l'=1}^L (a_{l,k}^{l'})^{\ast}  b_{l'',k}^{l'} \right|^2.
	\end{split}
	\end{equation}
	The variance of the beamforming gain uncertainty, $\Exp \{ |\mathtt{BU}_{l,k}|^2\}$, is evaluated as
	\begin{multline} \label{eq:BeamUncertaintyMMSE}
	\sum_{l'=1}^L p_{l',k}   \sum_{l''=1}^L |a_{l,k}^{l''}|^2   \Biggl( \underbrace{\Exp \left\{ |(\hat{\mathbf{h}}_{l'',k}^{l''})^\conjtr \mathbf{h}_{l',k}^{l''} |^2 \right\}}_{\mathcal{I}_1} \\
	- \underbrace{\left|\Exp \{(\hat{\mathbf{h}}_{l'',k}^{l''})^\conjtr \mathbf{h}_{l',k}^{l''} \} \right|^2}_{\mathcal{I}_2} \Biggr),
	\end{multline}
	where the expectation $\mathcal{I}_1$ is computed by applying the property in Lemma~\ref{Lemma:Supplementary} as
	\begin{equation} \label{eq:I1}
	\begin{split}
	\mathcal{I}_1 =& \Exp \left\{ |(\hat{\mathbf{h}}_{l'',k}^{l''})^\conjtr \hat{\mathbf{h}}_{l',k}^{l''} |^2 \right\} + \Exp \left\{ |(\hat{\mathbf{h}}_{l'',k}^{l''})^\conjtr \mathbf{e}_{l',k}^{l''} |^2 \right\} \\
	=& \hat{p}_{l'',k} \hat{p}_{l',k} \Exp \left\{ | \mathbf{y}_{l'',k}^\conjtr \pmb{\Psi}_{l'',k}^{-1} \mathbf{R}_{l'',k}^{l''} \mathbf{R}_{l',k}^{l''} \pmb{\Psi}_{l'',k}^{-1} \mathbf{y}_{l'',k}|^2 \right.  \\ &\left. + \mathrm{tr} \left(\Exp \{ \mathbf{e}_{l',k}^{l''} (\mathbf{e}_{l',k}^{l''})^\conjtr \hat{\mathbf{h}}_{l'',k}^{l''} (\hat{\mathbf{h}}_{l'',k}^{l''})^\conjtr\} \right) \right\}\\
	=& \tau_\text{p}^2 \hat{p}_{l'',k} \hat{p}_{l',k} |\mathrm{tr} (\pmb{\Psi}_{l'',k}^{-1} \mathbf{R}_{l'',k}^{l''} \mathbf{R}_{l',k}^{l''})|^2  \\
	&+ \tau_{\text{p}} \hat{p}_{l'',k} \mathrm{tr} (\mathbf{R}_{l',k}^{l''} \mathbf{R}_{l'',k}^{l''} \pmb{\Psi}_{l'',k}^{-1} \mathbf{R}_{l'',k}^{l''})\\
	=& \tau_\text{p} \left((b_{l',k}^{l''})^2 + c_{l'',k}^{l',k} \right)
	\end{split}
	\end{equation}
	and the expectation $\mathcal{I}_2$ is computed as
	\begin{equation}\label{eq:I2}
	\begin{split}
	\mathcal{I}_2 &= \left| \Exp \{ (\hat{\mathbf{h}}_{l'',k}^{l''} )^\conjtr \hat{\mathbf{h}}_{l',k}^{l''} \} \right|^2 \\
	&= \hat{p}_{l'',k} \hat{p}_{l',k} \left| \Exp \{ \mathbf{y}_{l'',k}^\conjtr  \pmb{\Psi}_{l'',k}^{-1} \mathbf{R}_{l'',k}^{l''} \mathbf{R}_{l',k}^{l''} \pmb{\Psi}_{l'',k}^{-1} \mathbf{y}_{l'',k} \} \right|^2\\
	&= \tau_{\text{p}}^2 \hat{p}_{l'',k} \hat{p}_{l',k} |\mathrm{tr} (\pmb{\Psi}_{l'',k}^{-1} \mathbf{R}_{l'',k}^{l''} \mathbf{R}_{l',k}^{l''})|^2 \\
	&= \tau_{\text{p}} (b_{l',k}^{l''})^2.
	\end{split}
	\end{equation}
	Combining \eqref{eq:BeamUncertaintyMMSE}, \eqref{eq:I1}, and \eqref{eq:I2}, we obtain the variance of the beamforming gain uncertainty as
	\begin{multline} \label{eq:MMSEBeamUncertain}
	\sum_{l''=1}^L \sum_{l'=1}^L p_{l'',k} \tau_{\text{p}} \hat{p}_{l',k} |a_{l,k}^{l'}|^2 \mathrm{tr} (\mathbf{R}_{l'',k}^{l'} \mathbf{R}_{l',k}^{l'}\pmb{\Psi}_{l',k}^{-1} \mathbf{R}_{l',k}^{l'} ) \\
	= \sum_{l''=1}^L   \sum_{l'=1}^L \tau_{\text{p}}  p_{l'',k} |a_{l,k}^{l'}|^2 c_{l',k}^{l'',k} .
	\end{multline}
	The variance of the non-coherent interference term, $\sum_{l'=1}^L \sum_{\substack{k'=1 \\ k' \neq k}}^K \Exp \{ |\mathtt{NI}_{l',k'}|^2\}$, is computed based on the independent channel properties  
	\begin{equation} \label{eq:MMSENonCoh}
	\begin{split} 
	&\sum_{l'=1}^L \sum_{\substack{k'=1 \\ k'\neq k}}^K p_{l',k'}  \sum_{l''=1}^L |a_{l,k}^{l''}|^2  \Exp \left\{\left|(\hat{\mathbf{h}}_{l'',k}^{l''})^\conjtr\mathbf{h}_{l',k'}^{l''} \right|^2 \right\} \\
	&= \sum_{l'=1}^L \sum_{\substack{k'=1 \\ k'\neq k}}^K \sum_{l''=1}^L p_{l',k'} \tau_{\text{p}} \hat{p}_{l'',k}  |a_{l,k}^{l''}|^2  \textrm{tr} \left( \mathbf{R}_{l',k'}^{l''} \mathbf{R}_{l'',k}^{l''}  \pmb{\Psi}_{l'',k}^{-1} \mathbf{R}_{l'',k}^{l''} \right) \\
	&= \sum_{l'=1}^L \sum_{\substack{k'=1 \\ k'\neq k}}^K \sum_{l''=1}^L\tau_{\text{p}}  p_{l',k'}  |a_{l,k}^{l''}|^2 c_{l'',k}^{l',k'}.
	\end{split}
	\end{equation}
	The last expectation in the denominator is computed as
	\begin{equation} \label{eq:MMSENoise}
	\begin{split}
	\Exp \{ |\mathtt{AN}|_{l,k}^2 \} &= \sum_{l'=1}^L |a_{l,k}^{l'}|^2  \Exp \{ | (\hat{\mathbf{h}}_{l',k}^{l'})^\conjtr \mathbf{n}_{l'} |^2 \} \\
	&= \sum_{l'=1}^L |a_{l,k}^{l'}|^2 \sigma^2 \tau_{\text{p}} \hat{p}_{l',k} \textrm{tr}   \left( \mathbf{R}_{l',k}^{l'} \pmb{\Psi}_{l',k}^{-1}  \mathbf{R}_{l',k}^{l'} \right) \\
	&= \sum_{l'=1}^L \tau_{\text{p}} |a_{l,k}^{l'}|^2 d_{l',k}. 
	\end{split}
	\end{equation}
	Using \eqref{eq:SENumberatorMMSE}, \eqref{eq:MMSEPilotContamination}, \eqref{eq:MMSEBeamUncertain}, \eqref{eq:MMSENonCoh}, and \eqref{eq:MMSENoise} in \eqref{eq:GeneralSE}, we obtain the closed-form expression for the SE as shown in Theorem~\ref{Theorem1}.
	
	\section{Proof of Theorem~\ref{Theorem1} in Case of EW-MMSE} \label{Appendix: Proof-Theorem-2}
	The main steps to prove the results for the case of \abb{EW-MMSE} are similar to that of \abb{MMSE}, but the distributions of the estimate and estimation errors are different (and not independent). However, we can use the relationship in Corollary~\ref{Corollary1} between the channels of the users sending non-orthogonal pilot signals to perform the derivation.  
	The main steps of the proof are summarized as follows: The numerator of \eqref{eq:GeneralSINR} is computed based on the relationship between the estimates of the channels to \abb{BS} $l$ and users~$k$ in cells~$l'$ and $l$ as stated in Corollary~\ref{Corollary1}:
	\begin{equation}\label{eq:EMMSENumerator}
	\begin{split}
	&p_{l,k} \left| \sum_{l'=1}^L (a_{l,k}^{l'})^{\ast} \frac{\sqrt{\hat{p}_{l,k}} \beta_{l,k}^{l'} }{\sqrt{\hat{p}_{l',k}} \beta_{l',k}^{l'} } \Exp \{ \|\hat{\mathbf{h}}_{l',k}^{l'} \|^2 \} \right|^2 \\
	&\quad= p_{l,k} \left| \sum_{l'=1}^L (a_{l,k}^{l'})^{\ast} \sqrt{\tau_{\text{p}}} \varrho_{l',k}^{l'} \varrho_{l,k}^{l'}
	\textrm{tr} \left( \pmb{ \Psi}_{l',k} \right)\right|^2 \\
	&\quad= p_{l,k} \left| \sum_{l'=1}^L (a_{l,k}^{l'})^{\ast} b_{l,k}^{l'} \right|^2.
	\end{split}
	\end{equation}
	Similarly, we use the relationship between the channel estimates of users~$k$ in cells~$l'$ and $l''$ to compute the variance of the pilot contamination term in the denominator of \eqref{eq:GeneralSINR} as
	\begin{equation} \label{eq:EMMSEPilotCont}
	\begin{split}
	&\sum_{\substack{l''=1 \\ l'' \neq l }}^L p_{l'',k} \left| \sum_{l'=1}^L (a_{l,k}^{l'})^{\ast} \frac{\sqrt{\hat{p}_{l'',k}} \beta_{l'',k}^{l'} }{ \sqrt{\hat{p}_{l',k}} \beta_{l',k}^{l'} } \Exp \{ \| \hat{\mathbf{h}}_{l',k}^{l'} \|^2 \}  \right|^2 \\
	&= \sum_{\substack{l''=1 \\ l'' \neq l }}^L p_{l'',k}  \left| \sum_{l'=1}^L (a_{l,k}^{l'})^{\ast} \sqrt{\tau_{\text{p}}} \varrho_{l',k}^{l'} \varrho_{l'',k}^{l'}
	\textrm{tr} \left( \pmb{ \Psi}_{l',k} \right) \right|^2 \\
	& = \sum_{\substack{l''=1 \\ l'' \neq l }}^L p_{l'',k}  \left| \sum_{l'=1}^L (a_{l,k}^{l'})^{\ast} b_{l'',k}^{l'} \right|^2.
	\end{split}
	\end{equation}
	The variance of the beam uncertainty term in the denominator of \eqref{eq:GeneralSINR} is computed as
	\begin{multline} \label{eq:EMMSEBeamUncertain}
	\sum_{l'=1}^L p_{l',k}   \sum_{l''=1}^L |a_{l,k}^{l''}|^2   \left( \Exp \left\{ |(\hat{\mathbf{h}}_{l'',k}^{l''})^\conjtr \hat{\mathbf{h}}_{l',k}^{l''} |^2 \right\} \right.\\\left.
	- \left(\Exp \{(\hat{\mathbf{h}}_{l'',k}^{l''})^\conjtr \hat{\mathbf{h}}_{l',k}^{l''} \} \right)^2 \right)  \\
	+ \sum_{l'=1}^L p_{l',k} \Exp \left\{ \left| \sum_{l''=1}^L (a_{l,k}^{l''})^{\ast}  (\hat{\mathbf{h}}_{l'',k}^{l''})^\conjtr \mathbf{e}_{l',k}^{l''}   \right|^2 \right\}.
	\end{multline}
	By using the relationship between the two channel estimates $\hat{\mathbf{h}}_{l'',k}^{l''}$ and $\hat{\mathbf{h}}_{l',k}^{l''}$ in Corollary~\ref{Corollary1}, \eqref{eq:EMMSEBeamUncertain} is equal to
	\begin{equation} \label{eq:EMMSEBeamUncertainv1}
	\begin{split}
	& \sum_{l'=1}^L p_{l',k} \sum_{l''=1}^L |a_{l,k}^{l''}|^2  \frac{\hat{p}_{l',k} (\beta_{l',k}^{l''})^2 }{\hat{p}_{l'',k} (\beta_{l'',k}^{l''})^2} \left( \Exp \{ \|\hat{\mathbf{h}}_{l'',k}^{l''} \|^4 \} \right. \\ 
	& \left. \quad-\left(\Exp \{ 
	\|\hat{\mathbf{h}}_{l'',k}^{l''} \|^2 \} \right)^2 \right) \\
	&\quad+ \sum_{l'=1}^L \sum_{l''=1}^L p_{l',k} \left|  a_{l,k}^{l''} \right|^2 \Exp \left\{ \left| (\hat{\mathbf{h}}_{l'',k}^{l''})^\conjtr \mathbf{e}_{l',k}^{l''} \right|^2 \right\}\\
	&= \sum_{l'=1}^L \sum_{l''=1}^L p_{l',k} \tau_{\text{p}} |a_{l,k}^{l''}|^2 (\varrho_{l'',k}^{l''})^2  (\varrho_{l',k}^{l''})^2 \textrm{tr} \left( \pmb{\Psi}_{l'',k}^2 \right) \\
	&\quad+ \sum_{l'=1}^L \sum_{l''=1}^L p_{l',k} | a_{l,k}^{l''} |^2 (\varrho_{l'',k}^{l''})^2 \\
	&\quad\times \textrm{tr} \left( \left(\mathbf{R}_{l',k}^{l''} - \tau_{\text{p}} (\varrho_{l',k}^{l''})^2 \pmb{\Psi}_{l'',k} \right) \pmb{\Psi}_{l'',k} \right)\\
	&= \sum_{l'=1}^L \sum_{l''=1}^L p_{l',k}  |a_{l,k}^{l''} |^2 (\varrho_{l'',k}^{l''})^2 \textrm{tr} \left( \mathbf{R}_{l',k}^{l''} \pmb{\Psi}_{l'',k}   \right)\\
	&= \sum_{l'=1}^L \sum_{l''=1}^L p_{l',k}  |a_{l,k}^{l''} |^2 c_{l'',k}^{l',k}.
	\end{split}
	\end{equation}
	By performing \abb{MMSE} estimation separately for every element of a channel vector, it is straightforward to prove that $\hat{\mathbf{h}}_{l'',k}^{l''}$ and $\mathbf{h}_{l',k'}^{l''}$ are independent since the joint density function is the product of their respective marginal densities. Consequently, the variance of the non-coherent interference in the denominator of \eqref{eq:GeneralSINR}, $\sum_{l'=1}^L \sum_{\substack{k'=1, k' \neq k}}^K \Exp \{ |\mathtt{NI}_{l',k'}|^2\}$, is computed as
	\begin{equation} \label{eq:EMMSENonCoh}
	\begin{split}
	& \sum_{l'=1}^L \sum_{\substack{k' =1 \\ k' \neq k}}^K p_{l',k'} \Exp \left\{ \left| \sum_{l''=1}^L (a_{l,k}^{l''})^{\ast}  (\hat{\mathbf{h}}_{l'',k}^{l''})^\conjtr\mathbf{h}_{l',k'}^{l''} \right|^2 \right\} \\
	& = \sum_{l'=1}^L \sum_{\substack{k'=1 \\ k'\neq k}}^K p_{l',k'}  \sum_{l''=1}^L |a_{l,k}^{l''}|^2  \Exp \left\{\left|(\hat{\mathbf{h}}_{l'',k}^{l''})^\conjtr\mathbf{h}_{l',k'}^{l''} \right|^2 \right\} \\
	& = \sum_{l'=1}^L \sum_{\substack{k'=1 \\ k'\neq k}}^K \sum_{l''=1}^L p_{l',k'}   |a_{l,k}^{l''}|^2  (\varrho_{l'',k}^{l''})^2 \textrm{tr} \left( \mathbf{R}_{l',k'}^{l''}  \pmb{\Psi}_{l'',k} \right) \\
	& = \sum_{l'=1}^L \sum_{\substack{k'=1 \\ k'\neq k}}^K \sum_{l''=1}^L p_{l',k'} |a_{l,k}^{l''} |^2 c_{l'',k}^{l',k'}.
	\end{split}
	\end{equation}
	The last expectation in the denominator of \eqref{eq:GeneralSE} is computed by using the fact that the noise and the channel estimate are independent, leading to
	\begin{equation} \label{eq:EMMSENoise}
	\begin{split}
	\Exp \{ |\mathtt{AN}_{l,k}|^2 \} & 
	=  \sum_{l'=1}^L |a_{l,k}^{l'}|^2 (\varrho_{l',k}^{l'} )^2 \sigma^2 \textrm{tr}   \left( \pmb{\Psi}_{l',k}  \right)\\
	&= \sum_{l'=1}^L |a_{l,k}^{l'}|^2 d_{l',k} .
	\end{split}
	\end{equation}
	Applying \eqref{eq:EMMSENumerator}, \eqref{eq:EMMSEPilotCont}, \eqref{eq:EMMSEBeamUncertainv1}, \eqref{eq:EMMSENonCoh}, and \eqref{eq:EMMSENoise} in \eqref{eq:GeneralSE}, we obtain the closed-form expression of the \abb{SE} as shown in Theorem~\ref{Theorem1}.
	
	\section{Proof of Theorem~\ref{Theorem:MMSEOptProblem}} \label{Appendix: Proof-Theorem-MMSEOptProblem}
	The whole system with the aggregate effect of channel and decoding can be viewed as a \abb{SISO} channel with deterministic channel gain, whose \abb{SE} is the equivalent of \eqref{eq: ClosedForm_Rate_MMSE}, namely:
	\begin{multline}
	y_{l,k} = \sqrt{p_{l,k}} \left( \sum_{l'=1}^L (a_{l,k}^{l'})^{\ast} b_{l,k}^{l'} \right) s_{l,k} \\
	+ \sum\limits_{\substack{l'=1 \\ l' \neq l }}^L \sqrt{p_{l',k}} \left( \sum\limits_{l''=1}^L (a_{l,k}^{l''})^{\ast} b_{l',k}^{l''} \right) s_{l',k} \\
	+ \sum\limits_{l'=1}^L \sum\limits_{k'=1}^K \sqrt{\sum\limits_{l''=1}^L p_{l',k'} | a_{l,k}^{l''}|^2 c_{l'',k}^{l',k'}} z_{l',k'} \\
	+ \sqrt{\sum\limits_{l'=1}^L | a_{l,k}^{l'}|^2 d_{l',k}} n_{l,k},
	\end{multline} 
	where $z_{l',k'} \sim \mathcal{CN}(0,1)$ and  $n_{l,k} \sim \mathcal{CN}(0,1)$. The desired signal $s_{l,k}$ is decoded by using a beamforming coefficient $u_{l,k} \in \mathbb{C}$ as
	\begin{multline} \label{eq:Sol_s_lkV1}
	\hat{s}_{l,k} = u_{l,k}y_{l,k} = \sqrt{p_{l,k}} u_{l,k} \left( \sum_{l'=1}^L (a_{l,k}^{l'})^{\ast} b_{l,k}^{l'} \right) s_{l,k} \\
	+ \sum\limits_{\substack{l'=1 \\ l' \neq l }}^L \sqrt{p_{l',k}} u_{l,k} \left( \sum\limits_{l''=1}^L (a_{l,k}^{l''})^{\ast} b_{l',k}^{l''} \right) s_{l',k} \\
	+ \sum\limits_{l'=1}^L \sum\limits_{k'=1}^K u_{l,k} \sqrt{\sum\limits_{l''=1}^L p_{l',k'} | a_{l,k}^{l''}|^2 c_{l'',k}^{l',k'}} z_{l',k'} \\ +  u_{l,k} \sqrt{\sum\limits_{l'=1}^L | a_{l,k}^{l'}|^2 d_{l',k}} n_{l,k}.
	\end{multline}
	We now compute the mean-square error as
	\begin{equation}
	\begin{split}
	&e_{l,k} = \Exp \{ |\hat{s}_{l,k} - s_{l,k}|^2 \} = \Exp \{ | u_{l,k}y_{l,k} - s_{l,k}|^2 \}  \\
	&=  \left|\sqrt{p_{l,k}} u_{l,k} \left( \sum_{l'=1}^L (a_{l,k}^{l'})^{\ast} b_{l,k}^{l'} \right) - 1 \right|^2 + \sum\limits_{\substack{l'=1 \\ l' \neq l }}^L p_{l',k} |u_{l,k}|^2 \times \\
	& \left| \sum\limits_{l''=1}^L  (a_{l,k}^{l''})^{\ast} b_{l',k}^{l''} \right|^2 +  \sum\limits_{l'=1}^L \sum\limits_{k'=1}^K \sum\limits_{l''=1}^L |u_{l,k}|^2 p_{l',k'} |a_{l,k}^{l''}|^2 c_{l'',k}^{l',k'} \\
	&+ |u_{l,k}|^2 \sum\limits_{l'=1}^L | a_{l,k}^{l'}|^2 d_{l',k}.
	\end{split}
	\end{equation}
	After some algebra, we obtain $e_{l,k}$ as shown in \eqref{eq: errorOpt}. The optimal solution of $u_{l,k}$ is computed by equating the first derivative of $e_{l,k}$ with respect to $u_{l,k}$ to zero, leading to
	\begin{multline}
	u_{l,k}^{\ast} \left( \sum\limits_{l'=1}^L p_{l',k} \left| \sum\limits_{l''=1}^L (a_{l,k}^{l''})^{\ast} b_{l',k}^{l''} \right|^2\right.\\\left.
	+ \sum\limits_{l'=1}^L \sum\limits_{k'=1}^K \sum\limits_{l''=1}^L p_{l',k'} |a_{l,k}^{l''}|^2 c_{l'',k}^{l',k'} \right. \\ \left. + \sum\limits_{l'=1}^L | a_{l,k}^{l'}|^2 d_{l',k}  \right) 
	- \sqrt{p_{l,k}} \sum_{l'=1}^L (a_{l,k}^{l'})^{\ast} b_{l,k}^{l'} = 0.
	\end{multline}
	Therefore, the optimal solution $u_{l,k}^{\mathrm{opt}}$ for a given set $\{\mathbf{a}_{l,k}, w_{l,k}, \rho_{l,k} \}$ is given in \eqref{eq:Opt_u_lk}.
	\begin{figure*}
		\begin{equation} \label{eq:Opt_u_lk}
		u_{l,k}^{\mathrm{opt}}= \frac{\sqrt{p_{l,k}}  \sum_{l'=1}^L a_{l,k}^{l'} (b_{l,k}^{l'})^{\ast} }{ \sum_{l'=1 }^L p_{l',k} \left| \sum_{l''=1}^L  (a_{l,k}^{l''})^{\ast} b_{l',k}^{l''} \right|^2 + \sum_{l'=1}^L \sum_{k'=1}^K \sum_{l''=1}^L p_{l',k'} | a_{l,k}^{l''}|^2 c_{l'',k}^{l',k'} + \sum_{l'=1}^L | a_{l,k}^{l'}|^2 d_{l',k} } .
		\end{equation} \hrulefill
		%\vspace*{-0.5cm}
	\end{figure*}
	
	The optimal solution $w_{l,k}^{\mathrm{opt}}$ is obtained by taking the first derivative of the objective function of the optimization problem \eqref{Problem: Sumratev2} with respect to $w_{l,k}$ and equating to zero:
	\begin{equation} \label{eq: Opt_wlk}
	w_{l,k}^{\mathrm{opt}} = e_{l,k}^{-1}.
	\end{equation} 
	Using \eqref{eq:Opt_u_lk} and \eqref{eq: Opt_wlk} in \eqref{Problem: Sumratev2}, we obtain the optimization problem
	\begin{equation} \label{Problem: Sumratev4}
	\begin{aligned}
	& \underset{ \{ p_{l,k} \geq 0 \}, \{ \mathbf{a}_{l,k} \} }{\mathrm{minimize}}
	&&  KL -  \sum_{l=1}^L \sum_{k=1}^K \ln \left(1 + \mathrm{SINR}_{l,k} \right) \\
	& \mathrm{subject\,\,to}
	& & p_{l,k} \leq P_{\mathrm{max},l,k} \;, \forall l,k.
	\end{aligned}
	\end{equation}
	Since \eqref{Problem: Sumratev4} is easily converted to \eqref{Problem: Sumratev1}, we have completed the proof.
	
	\section{Proof of Theorem~\ref{Theorem: IterativeSol}} \label{Appendix: Proof_Of_Theorem: IterativeSol}
	For sake of simplicity, we omit the iteration index in the proof. The optimal solution of $u_{l,k}$ and $w_{l,k}$ are easily computed by \eqref{eq:Opt_u_lk} and \eqref{eq: Opt_wlk} by noting that $\rho_{l,k} = \sqrt{p_{l,k}}$. We can find the optimal solution to $\mathbf{a}_{l,k}$ for a given set of $\{ u_{l,k}, w_{l,k}, \rho_{l,k}\}$ from the  optimization problem
	\begin{equation} \label{Problem: a_lk_Subproblem}
	\begin{aligned}
	& \underset{\{ \mathbf{a}_{l,k} \} }{\mathrm{minimize}}
	&&  \sum_{l=1}^L \sum_{k=1}^K w_{l,k} \tilde{e}_{l,k}
	\end{aligned}
	\end{equation}
	where $\tilde{e}_{l,k}$ in \eqref{Problem: a_lk_Subproblem} depends on $\{ \mathbf{a}_{l,k} \} $ and is defined as
	\begin{multline} \label{eq: errorOptv5}
	\tilde{e}_{l,k} = |u_{l,k}|^2 \left( \sum\limits_{l'=1 }^L \rho_{l',k}^2 \left| \sum\limits_{l''=1}^L  (a_{l,k}^{l''})^{\ast} b_{l',k}^{l''} \right|^2  \right.\\\left.
	+ \sum\limits_{l'=1}^L \sum\limits_{k'=1}^K \sum\limits_{l''=1}^L \rho_{l',k'}^2  |a_{l,k}^{l''}|^2 c_{l'',k}^{l',k'} + \sum\limits_{l'=1}^L | a_{l,k}^{l'}|^2 d_{l',k} \right) \\
	- \rho_{l,k}  u_{l,k}  \left( \sum_{l'=1}^L (a_{l,k}^{l'})^{\ast} b_{l,k}^{l'} \right) -
	\rho_{l,k}  u_{l,k}^{\ast}  \left( \sum_{l'=1}^L a_{l,k}^{l'} (b_{l,k}^{l'})^{\ast} \right).
	\end{multline}
	By denoting $f(\mathbf{a}_{l,k}) = \sum_{l=1}^L \sum_{k=1}^K w_{l,k} \tilde{e}_{l,k}$ and using the expression of $\tilde{e}_{l,k}$ in \eqref{eq: errorOptv5}, we can write $f(\mathbf{a}_{l,k})$ as
	\begin{multline} \label{eq: f_function}
	f(\mathbf{a}_{l,k}) = \sum_{l=1}^L \sum_{k=1}^K w_{l,k} \left( |u_{l,k}|^2 \mathbf{a}_{l,k}^\conjtr \widetilde{\mathbf{C}}_{l,k} \mathbf{a}_{l,k} \right. \\ \left. - u_{l,k} \rho_{l,k} \mathbf{a}_{l,k}^\conjtr \mathbf{b}_{l,k}  - u_{l,k}^{\ast} \rho_{l,k} \mathbf{b}_{l,k}^\conjtr \mathbf{a}_{l,k} \right).
	\end{multline}
	Taking the first derivative of $f(\mathbf{a}_{l,k})$ with respect to $\mathbf{a}_{l,k}$, we obtain
	\begin{equation}
	\nabla f  = 2w_{l,k} |u_{l,k}|^2 \widetilde{\mathbf{C}}_{l,k} \mathbf{a}_{l,k} -  2 w_{l,k} u_{l,k}  \rho_{l,k} \mathbf{b}_{l,k}.
	\end{equation}
	Therefore, the solution is
	\begin{equation} \label{eq:alkv1}
	\mathbf{a}_{l,k}^{\mathrm{opt}} = \frac{\rho_{l,k}}{u_{l,k}^{\ast}} \widetilde{\mathbf{C}}_{l,k}^{-1} \mathbf{b}_{l,k}.
	\end{equation}
	After removing $\rho_{l,k}$ in both the numerator and denominator of the fraction in \eqref{eq:alkv1} and doing some algebra, the optimal solution to $\mathbf{a}_{l,k}$ is expressed as in \eqref{eq: a_lk_n1}.
	
	We now compute the optimal solution for $\rho_{l,k}$ for a given set of optimization variables $\{ \mathbf{a}_{l,k}, w_{l,k}, u_{l,k} \}$. In this case, \eqref{Problem: Sumratev3} simplifies to
	\begin{equation} \label{Problem: PowerSubproblem}
	\begin{aligned}
	& \underset{\{ \rho_{l,k} \geq 0 \} }{\mathrm{minimize}}
	&&  \sum_{l=1}^L \sum_{k=1}^K w_{l,k} e_{l,k} \\
	& \mathrm{subject\,\,to}
	& & \rho_{l,k}^2 \leq P_{\mathrm{max},l,k} \;, \forall l,k. \\
	\end{aligned}
	\end{equation} 
	The Lagrangian function of the optimization \eqref{Problem: PowerSubproblem} is 
	\begin{multline}
	\mathcal{L} \left( \{\rho_{l,k}\}, \{\lambda_{l,k}\}\right) = \sum_{l=1}^L \sum_{k=1}^K w_{l,k} e_{l,k} + \\ \sum_{l=1}^L \sum_{k=1}^K \lambda_{l,k} \left(\rho_{l,k}^2 - P_{\max,l,k} \right),
	\end{multline}
	where $\lambda_{l,k} \geq 0$ is the Lagrange multiplier associated with $\rho_{l,k}^2 \leq P_{\mathrm{max},l,k}$. 
	Taking the first derivative of $\mathcal{L} \left( \{\rho_{l,k}\}, \{\lambda_{l,k}\}\right)$ with respect to $\rho_{l,k}$, we obtain
	\begin{multline}
	\frac{\partial \mathcal{L} }{ \partial \rho_{l,k}} = 2\rho_{l,k} \sum_{l'=1}^L w_{l',k} |u_{l',k}|^2 \left| \sum_{l''=1}^L (a_{l',k}^{l''})^{\ast} b_{l,k}^{l''} \right|^2\\
	+ 2\rho_{l,k} \sum_{l'=1}^L \sum_{k'=1}^K w_{l',k'} |u_{l',k'}|^2 \sum_{l''=1}^L | a_{l',k'}^{l''}|^2 c_{l'',k'}^{l,k} \\
	- w_{l,k} u_{l,k}  \left( \sum_{l'=1}^L (a_{l,k}^{l'})^{\ast} b_{l,k}^{l'} \right) - w_{l,k} u_{l,k}^{\ast} \left( \sum_{l'=1}^L a_{l,k}^{l'} (b_{l,k}^{l'})^{\ast} \right) \\
	+ 2\lambda_{l,k} \rho_{l,k}.
	\end{multline}
	By equating the above derivative to zero, the stationary point is obtained as shown in \eqref{eq:rholk}. 
	\begin{figure*}
		\begin{equation} \label{eq:rholk}
		\rho_{l,k} = \frac{ w_{l,k} \mathfrak{Re} \left(u_{l,k} \sum_{l'=1}^L (a_{l,k}^{l'})^{\ast} b_{l,k}^{l'} \right) }{\sum_{l'=1}^L w_{l',k} |u_{l',k}|^2 \left| \sum_{l''=1}^L (a_{l',k}^{l''})^{\ast} b_{l,k}^{l''} \right|^2 + \sum_{l'=1}^L \sum_{k'=1}^K w_{l',k'} |u_{l',k'}|^2 \sum_{l''=1}^L | a_{l',k'}^{l''}|^2 c_{l'',k'}^{l,k} + \lambda_{l,k}}.
		\end{equation}%\vspace*{-0.5cm}
	\end{figure*}
	The Lagrangian multiplier $\lambda_{l,k}$ must satisfy the complementary slackness condition \cite{Boyd2004a}
	\begin{equation}
	\lambda_{l,k} \left( \rho_{l,k}^2 - P_{\max,l,k} \right) =0.
	\end{equation}
	Therefore, we obtain the solution to $\rho_{l,k}$ as 
	\begin{equation} \label{eq:rholkSolv1}
	\rho_{l,k} = \begin{cases}
	\min (\tilde{\rho}_{l,k}, \sqrt{P_{\max,l,k}}), & \mbox{ if } \lambda_{l,k} =0,\\ 
	\sqrt{P_{\max,l,k}}, &\mbox{ if } \lambda_{l,k} \neq 0,
	\end{cases}
	\end{equation}
	where $\tilde{\rho}_{l,k}$ is defined as in \eqref{eq:rholkSolv2},
	\begin{figure*}
		\begin{equation}\label{eq:rholkSolv2}
		\tilde{\rho}_{l,k} \triangleq \frac{ w_{l,k} \mathfrak{Re}\left(u_{l,k}  \sum_{l''=1}^L (a_{l,k}^{l''})^{\ast} b_{l,k}^{l''} \right) }{\sum_{l'=1}^L w_{l',k} |u_{l',k}|^2 \left| \sum_{l''=1}^L (a_{l',k}^{l''})^{\ast} b_{l,k}^{l''} \right|^2 + \sum_{l'=1}^L \sum_{k'=1}^K w_{l',k'} |u_{l',k'}|^2 \sum_{l''=1}^L | a_{l',k'}^{l''}|^2 c_{l'',k'}^{l,k}}.
		\end{equation}% \vspace*{-0.5cm}
		\hrulefill% \vspace*{-0.15cm}
	\end{figure*}
	which is obtained from \eqref{eq:rholk} by setting $\lambda_{l,k} = 0$. From \eqref{eq:rholkSolv1}, the optimal solution to $\rho_{l,k}$ is derived as shown in \eqref{eq:rho_lkSol}.
	
	We now prove that Algorithm~\ref{Algorithm:CentralizedApproach} converges to a stationary point, as defined in Definition~\ref{Def:StationaryPoint}. The optimization problem \eqref{Problem: Sumratev3} is first converted to the following equivalent unconstrained problem:
	\begin{equation} \label{Problem: Sumratev6}
	\begin{aligned}
	& \underset{ \substack{\{ \rho_{l,k} \geq 0 \}, \{ \mathbf{a}_{l,k} \},\\ \{ w_{l,k} \geq 0 \}, \{ u_{l,k} \} }}{\mathrm{minimize}}
	&&  g (\{u_{l,k} \}, \{ w_{l,k} \}, \{\mathbf{a}_{l,k} \},\{\rho_{l,k}\})
	\end{aligned}
	\end{equation}
	where the objective function $g$ is defined as
	\begin{equation}
	\begin{split}
	&g (\{u_{l,k} \}, \{ w_{l,k} \}, \{\mathbf{a}_{l,k} \},\{\rho_{l,k}\}) =\\
	& \sum_{l=1}^L \sum_{k=1}^K w_{l,k} e_{l,k} - \ln (w_{l,k}) + \lambda_{l,k}\left( \rho_{l,k}^2 - P_{\max,l,k} \right).
	\end{split}
	\end{equation}
	Since in every iteration, each subproblem is convex and has a unique optimal solution,
	the objective function of the optimization problem \eqref{Problem: Sumratev6} is monotonically decreasing after iterations.
	Additionally, this function is lower bounded by zero, so  Algorithm~\ref{Algorithm:CentralizedApproach} must converge to a limit point, attained by a solution that we call $ \{(u_{l,k}^{ \mathrm{opt}}, w_{l,k}^{ \mathrm{opt}},  \mathbf{a}_{l,k}^{ \mathrm{opt}}, \rho_{l,k}^{ \mathrm{opt}})\}$. 
	Note that $g (\{u_{l,k} \}, \{ w_{l,k} \}, \{\mathbf{a}_{l,k} \},\{\rho_{l,k}\})$ is convex in each optimization variable, when the others variables are fixed, and the optimal solution to each sub-problem is computed from the first derivative of the cost function. By applying the standard trick in  \cite[Remark~2.2]{Bjornson2013d} to decompose a complex number into the real and imaginary parts, the following properties are obtained:
	\begin{align}
	& \left( \mathfrak{Re}(u_{l,k}) - \mathfrak{Re}(u_{l,k}^{\mathrm{opt}}) \right)  \frac{\partial g}{\partial \mathfrak{Re}(u_{l,k})} \biggl|_{ \mathfrak{Re}(u_{l,k}^{\mathrm{opt}})} \geq 0,\\
	& \left( \mathfrak{Im}(u_{l,k}) - \mathfrak{Im}(u_{l,k}^{\mathrm{opt}}) \right)  \frac{\partial g}{\partial \mathfrak{Im}(u_{l,k})} \biggl|_{ \mathfrak{Im}(u_{l,k}^{\mathrm{opt}})} \geq 0, \\
	&  \left( w_{l,k} - w_{l,k}^{\mathrm{opt}} \right)  \frac{\partial g }{\partial w_{l,k}} \biggl|_{w_{l,k}^{\mathrm{opt}}} \geq 0,\\
	&  \left( \mathfrak{Re}(\mathbf{a}_{l,k}) - \mathfrak{Re}(\mathbf{a}_{l,k}^{\mathrm{opt}}) \right)^T \nabla g \biggl|_{\mathfrak{Re}(\mathbf{a}_{l,k}^{\mathrm{opt}})} \geq 0,\\
	&  \left( \mathfrak{Im}(\mathbf{a}_{l,k}) - \mathfrak{Im}(\mathbf{a}_{l,k}^{\mathrm{opt}}) \right)^T \nabla g \biggl|_{\mathfrak{Im}(\mathbf{a}_{l,k}^{\mathrm{opt}})} \geq 0,\\
	& \left( \rho_{l,k} - \rho_{l,k}^{\mathrm{opt}} \right) \frac{\partial g }{\partial \rho_{l,k}} \biggl|_{ \rho_{l,k}^{\mathrm{opt}}} \geq 0.
	\end{align}
	These properties mean that the limit point is a stationary point to \eqref{Problem: Sumratev3}.
	
	We now prove that the optimal solution $\{\mathbf{a}_{l,k}^{\mathrm{opt}}\}$, $\{(\rho_{l,k}^{\mathrm{opt}})^2\}$ forms a stationary point of \eqref{Problem: Sumratev1}. In fact, the optimization problem \eqref{Problem: Sumratev1} is equivalent to 
	\begin{equation} \label{Problem: Sumraterho}
	\begin{aligned}
	& \underset{ \{ \rho_{l,k} \geq 0 \},  \{ \mathbf{a}_{l,k} \}  }{\textrm{maximize}}
	&&  h(\{\mathbf{a}_{l,k} \}, \{\rho_{l,k}\})\\
	& \mathrm{subject\,\,to}
	& & \rho_{l,k}^2 \leq P_{\mathrm{max},l,k} \;, \forall l,k. \\
	\end{aligned}
	\end{equation}
	where the objective function is
	\begin{equation}
	\begin{split}
	h(\{\mathbf{a}_{l,k} \}, \{\rho_{l,k}\}) &\triangleq \sum_{l=1}^{L} \sum_{k=1}^K \log_2 \left(1 + \mathrm{SINR}_{l,k}\right).
	\end{split}
	\end{equation}
	Here, the SINR value has a similar expression as in \eqref{eq:MMSE_SINR}, but with $p_{l,k} = \rho_{l,k}^2$. For given $w_{l,k} = w_{l,k}^{\mathrm{opt}}$ and $u_{l,k} = u_{l,k}^{\mathrm{opt}}$, for all $l,k$, it is sufficient to prove the following equalities
	\begin{align}
	\frac{\partial g  }{\partial \rho_{l',k'}} &= \frac{1}{\log_2(e)} \frac{ \partial h }{ \partial \rho_{l',k'} } + 2 \lambda_{l',k'} \rho_{l',k'}, \forall l',k', \label{eq:KKT1}\\
	\nabla g ({\mathbf{a}_{l',k'}}) &= \frac{1}{\log_2(e)} \nabla h (\mathbf{a}_{l',k'} ) , \forall l',k'. \label{eq:KKT2}
	\end{align}
	By using \eqref{eq:Opt_u_lk}, \eqref{eq: Opt_wlk}, and the chain rule, \eqref{eq:KKT1} is proved as
	\begin{equation}
	\begin{split}
	\frac{\partial g }{\partial \rho_{l',k'}} &= \sum_{l=1}^L \sum_{k=1}^K w_{l,k}^{\mathrm{opt}} \frac{ \partial e_{l,k}}{ \partial \rho_{l',k'}} + 2 \lambda_{l',k'} \rho_{l',k'} \\
	&= \sum_{l=1}^L \sum_{k=1}^K(e_{l,k}^{\mathrm{opt}})^{-1} \frac{ \partial e_{l,k}^{\mathrm{opt}}}{ \partial \rho_{l',k'}} + 2 \lambda_{l',k'} \rho_{l',k'},
	\end{split}
	\end{equation}
	where $e_{l,k}^{\mathrm{opt}} = \left(1+ \mathrm{SINR}_{l,k} \right)^{-1}$ is derived by using \eqref{eq:Opt_u_lk} in \eqref{eq:errorOptv2} and some algebra. 
	It leads to
	\begin{equation}
	\begin{split}
	\frac{\partial g  }{\partial \rho_{l',k'}} =& \sum_{l=1}^L \sum_{k=1}^K \left(1+ \mathrm{SINR}_{l,k} \right) \frac{ \partial \left(1+ \mathrm{SINR}_{l,k} \right)^{-1}}{ \partial \rho_{l',k'}} \\
	& + 2 \lambda_{l',k'} \rho_{l',k'}\\
	=& \sum_{l=1}^L \sum_{k=1}^K \left(1+ \mathrm{SINR}_{l,k} \right)^{-1} \frac{\partial \mathrm{SINR}_{l,k}}{\partial \rho_{l',k'}} + 2 \lambda_{l',k'} \rho_{l',k'} \\
	=& \frac{1}{\log_2(e)} \frac{ \partial h }{ \partial \rho_{l',k'}} + 2 \lambda_{l',k'} \rho_{l',k'}.
	\end{split}
	\end{equation}
	The proof of \eqref{eq:KKT2} is similar to how \eqref{eq:KKT1} just was proved.
	%======================== ==Reference==========================================================================
	\bibliographystyle{IEEEtran}
	\bibliography{IEEEabrv,refsv1}
	
	\begin{IEEEbiography} 
		%[{\includegraphics[width=1.0in,height=1.25in,clip,keepaspectratio]{illustrations/trinhvanchien.jpg}}]
		{Trinh Van Chien} (S'16) received the B.S. degree in electronics and telecommunications from the Hanoi University of Science and Technology, Vietnam,
		in 2012, and the M.S. degree in electrical and computer enginneering from Sungkyunkwan University, South Korea, in 2014. He is currently pursuing the Ph.D. degree in communication systems from Link\"oping University, Sweden and being a Marie Sklodowska-Curie research fellow at the H2020 5Gwireless ITN project. His interest lies in
		convex optimization problems for wireless communications and image and video processing. He was an IEEE Wireless Communications Letters Exemplary
		Reviewer for 2016 and 2017. He was a recipient of the Award of Scientific Excellence
		in the first year of the 5Gwireless project funded by European Union Horizon's 2020. 
	\end{IEEEbiography}
	
	\begin{IEEEbiography}%[{\includegraphics[width=1in,height=1.25in,clip,keepaspectratio]{chris.jpg}}]
		{Christopher~Mollén} received the M.Sc.\ and
		Ph.D.\ degrees from Linköping University, Sweden,
		in 2013 and 2018, respectively.  
		His Ph.D.\ thesis “High-End Performance with Low-End Hardware: Analysis of Massive MIMO Base Stations
		Transceivers” explored low-complexity hardware
		implementations of massive \abb{MIMO} base stations,
		including low-\abb{PAR} precoding, low-resolution \abb{ADC}s,
		and nonlinear amplifiers. 
		Previously, he was an intern at Ericsson in Kista, Sweden, and in Shanghai, China. 
		From 2011 to 2012, he was an exchange student in the Erasmus Programme at the Eidgenössische Technische Hochschule Zürich, Switzerland. 
		From 2015 to 2016, he visited the University of Texas at Austin as a Fulbright Scholar. 
		The IEEE Transactions on Communications named him exemplary reviewer in 2017.
		This work was done while at Linköping University. 
		He is currently with the 3D Vision Team at Apple.
	\end{IEEEbiography}
	
	\begin{IEEEbiography} 
		%[{\includegraphics[width=1.0in,height=1.25in,clip,keepaspectratio]{illustrations/emilbjornson.jpg}}]
		{Emil Bj\"ornson} (S'07-M'12-SM'17) received the M.S. degree in Engineering Mathematics from Lund University, Sweden, in 2007. He received the Ph.D. degree in Telecommunications from KTH Royal Institute of Technology, Sweden, in 2011. From 2012 to mid 2014, he was a joint postdoc at the Alcatel-Lucent Chair on Flexible Radio, SUPELEC, France, and at KTH. He joined Link\"oping University, Sweden, in 2014 and is currently Associate Professor and Docent at the Division of Communication Systems.
		
		He performs research on multi-antenna communications, Massive MIMO, radio resource allocation, energy-efficient communications, and network design. He is on the editorial board of the IEEE Transactions on Communications (since 2017) and the IEEE Transactions on Green Communications and Networking (since 2016). He is the first author of the textbooks ``Massive MIMO Networks: Spectral, Energy, and Hardware Efficiency'' (2017)  and ``Optimal Resource Allocation in Coordinated Multi-Cell Systems'' from 2013. He is dedicated to reproducible research and has made a large amount of simulation code publicly available.
		
		Dr. Bj\"ornson has performed MIMO research for more than ten years and has filed more than ten related patent applications. He received the 2018 Marconi Prize Paper Award in Wireless Communications, the 2016 Best PhD Award from EURASIP, the 2015 Ingvar Carlsson Award, and the 2014 Outstanding Young Researcher Award from IEEE ComSoc EMEA. He also co-authored papers that received best paper awards at the conferences WCSP 2017, IEEE ICC 2015, IEEE WCNC 2014, IEEE SAM 2014, IEEE CAMSAP 2011, and WCSP 2009.
	\end{IEEEbiography}
	
\end{document}